\documentclass[12pt]{article}
\UseRawInputEncoding
\usepackage{amssymb}
\usepackage{amsfonts}
\usepackage{mathtools}
\usepackage{bm}
\usepackage[all]{xy}
\usepackage{graphicx}
\usepackage{verbatim}
\usepackage{wrapfig}
\usepackage{makeidx}
\usepackage{amscd}
\usepackage{url}
\usepackage{comment}
\usepackage{enumerate}
\usepackage{here}
\usepackage{latexsym}
\usepackage{array}
\usepackage{booktabs}
\usepackage{multirow}
\usepackage{mathrsfs}
\usepackage{geometry}
\usepackage{fancyhdr}

\renewcommand{\title}[1]{

\begin{center} \Large \bf #1 \end{center}
}

\renewcommand{\author}[2]{
 \begin{center} #1  \vspace{3mm} \\
  #2 \\
 \end{center}
\addvspace{\baselineskip}
}

\usepackage{amssymb}
\usepackage{amsmath}

\usepackage{amsthm}
\newtheorem{theorem}{Theorem}[section]
\newtheorem{proposition}[theorem]{Proposition}

\newtheorem{lemma}[theorem]{Lemma}

\theoremstyle{definition}
\newtheorem{definition}[theorem]{Definition}

\theoremstyle{remark}
\newtheorem*{rem}{Remark}

\makeatletter
\@addtoreset{equation}{section}

\makeatother

\begin{document}

\baselineskip 5mm
\title{Quantization of Algebraic Varieties Defined by Casimir Polynomials via Matrix Regularization:\\
Fuzzy $S^7$ and Beyond}
\author{ Akifumi Sako}{
Department of Mathematics, Faculty of Science Division II, \\
 Tokyo University of Science,\\
 1-3 Kagurazaka, Shinjuku-ku, Tokyo, 162-8601, Japan
}
\noindent
%

\abstract{
We study the quantization of algebraic varieties defined by equations involving Casimir polynomials of compact semisimple Lie algebras. The Casimir polynomials belong to the Poisson center of the corresponding Lie-Poisson algebra. For this purpose, we employ a recently developed matrix regularization of Lie-Poisson algebras. In particular, using its formulation based on reducible representations, we construct quantizations of these algebraic varieties through their decomposition into coadjoint orbits, including singular orbits. As a concrete example, we present the construction of fuzzy $S^7$ in detail.
}

\section{ Introduction }

\subsection{Background}

This paper concerns the quantization of algebraic varieties defined by algebraic equations, with particular emphasis on matrix regularization. Since matrix regularization is closely related to physical theories known as matrix models, we begin with a brief introduction to matrix models.

In attempts to formulate string theory, a candidate theory of quantum gravity, in a nonperturbative manner, matrix models such as the type IIB matrix model, also known as the IKKT (Ishibashi--Kawai--Kitazawa--Tsuchiya) matrix model, and the BFSS (Banks--Fischler--Shenker--Susskind) matrix model naturally arise. Matrix algebras are used to describe these theories \cite{BFSS,IKKT3}. See also \cite{IKKT2,fuzzya} and the references therein. It is therefore important to understand how matrix algebras encode geometries such as the spacetime in which we live.

One approach to investigating the spacetime emerging from matrix models is numerical simulation. The emergence of a $(3+1)$-dimensional spacetime was reported in \cite{Kim:2011cr,Anagnostopoulos:2026qvz}, and many other results have subsequently been obtained \cite{Anagnostopoulos:2022dak}. Another approach is to construct geometry from the viewpoint of noncommutative geometry. An early form of this idea can be found in \cite{Gonzalez-Arroyo:1982hyq}. It was subsequently developed in \cite{IKKT2,Aoki:1999vr} and in studies of gauge theories on noncommutative spaces \cite{Ambjorn:1999ts,Ambjorn:2000nb,Ambjorn:2000cs}. Several reviews of these developments are available in \cite{fuzzya,Ydri:2017ncg}. The matrix regularization considered in this paper also treats matrix models from the viewpoint of noncommutative geometry.
Further new developments in these matrix models can be found, for example, in \cite{Chou:2025moy,Hattori:2026hrh,Asano:2026mal} and the references therein. 

The Lie algebras considered in this paper arise as classical solutions of the matrix models studied in \cite{Sako:2022pid,Gohara_Sako}. The model includes, as a special case, the bosonic sector of the mass-deformed IKKT matrix model. 
It is known that there are various classical solutions when the model is deformed with a mass term.
(See for example \cite{DiscreteMiniSurface,Kim:2011ts,Kim:2012mw,Sperling:2019xar}.)
A mass term in the IKKT matrix model was also discussed in connection with a new regularization that preserves Lorentz symmetry in \cite{Asano:2024def}. It has furthermore been suggested that an effective mass term may be dynamically generated \cite{Laliberte:2024iof}. 
The results of the present paper provide insight into the classical geometries and classical dynamics associated with matrix configurations that appear as classical solutions of the matrix models.\\
\bigskip

We next briefly review the history of matrix regularization and of the fuzzy spaces obtained from it. Matrix regularization is a sequence of quantizations that map an algebra of functions on a symplectic manifold to matrix algebras. This mapping provides an approximate correspondence between Poisson brackets and commutators.
The history of fuzzy spaces begins with the proposal of the fuzzy sphere \cite{matrix1,fuzzy1}. The fuzzy sphere provides a map from a set of functions on $S^2$ to the endomorphism algebra of an $N$-dimensional vector space. This matrix algebra is generated by an irreducible representation of ${su}(2)$. Further details can be found in \cite{matrix1,fuzzy1,fuzzyb,fuzzyc} and the references therein.
Fuzzy Riemann surfaces have also been constructed
\cite{hoppe1,hoppe2,fuzzyb,arnlind2,arnlind3,Schneiderbauer,klimek}.
The quantization of K\"{a}hler manifolds using Toeplitz operators
\cite{berezin1,klimek} can also be understood as a matrix regularization.
Fuzzy spaces have been studied using quasi-coherent states, and it has
been shown that their classical structures can be extracted without
taking the limit $N\to\infty$ \cite{Schneiderbauer,steinacker2}.
More recently, matrix regularizations based on Howe duality have also been
developed \cite{Hasebe:2026meg}.

The Lie-Poisson algebras considered in this paper are restricted to those
associated with compact groups. However, Lie algebras associated with
noncompact groups can also be considered, and several corresponding fuzzy
spaces are known
\cite{Hasebe:2012mz,Jurman:2013ota,Jurman:2017kkp,
Sperling:2018xrm,Sperling:2019xar}.
Although it is not possible to list all known examples here, surveys and
summaries can be found, for example, in
\cite{Bal:2004ai,Nair:1998bp,Grosse:1999ci,ABIY,
Balachandran:2001dd,Azuma_Bal_Nagao,Grosse:2004wm,
Steinacker_text,Rieffel:2021ykh}.

It is important to note that the commutative limit of a matrix
regularization does not necessarily coincide with the manifold from which
the quantization was originally constructed. When a classical geometry,
or equivalently a Poisson algebra, is extracted from a matrix algebra, the
resulting Poisson algebra depends on the choice of the classical
commutative limit
\cite{berezin1,bordemannA,Chu:2001xi}.
See also
\cite{deWitHoppeNicolai,arnlind,matrix1,fuzzy1,MadoreText}.
The problem of determining the appropriate classical commutative limit is
sometimes referred to as the inverse problem. Various approaches to
extracting geometric properties from matrix algebras have been studied
\cite{shimada,berenstein,Schneiderbauer,ishiki1,ishiki2,asakawa}.

\bigskip

In this paper, we study the matrix regularization constructed in
\cite{Sako:2022pid,Gohara_Sako} as a quantization of algebraic varieties
invariant under the action of semisimple Lie groups. The action of the Lie
algebra on an algebraic variety is described by a Lie-Poisson algebra,
and the variety is defined using one of its Casimir polynomials.

The Lie-Poisson algebra used in this paper is defined as follows.
\begin{definition}
Let $x=(x_1,x_2,\ldots,x_d)$ be a set of commuting variables, and suppose
that
$
\bigl(\mathbb{C}[x],\cdot,\{\ ,\ \}\bigr)
$
is a Poisson algebra. If the Poisson bracket is linear on the generators,
namely, if there exist structure constants $f_{ij}^{k}$ such that
\begin{align}\label{poisson_bra}
\{x_i,x_j\}=f_{ij}^{k}x_k,
\end{align}
then
$
\bigl(\mathbb{C}[x],\cdot,\{\ ,\ \}\bigr)
$
is called a Lie-Poisson algebra.
\end{definition}
As in Eq.~(\ref{poisson_bra}), throughout this paper we adopt the Einstein summation convention for all indices. When ambiguity may arise, we also use the explicit $\sum$ symbol.
\\

Lie-Poisson structures are discussed in
\cite{S_Lie,Weinstein1983}.
An important theory was developed through the geometric quantization of coadjoint orbits, 
establishing a correspondence between suitable coadjoint orbits of a Lie-Poisson algebra and irreducible representations of the corresponding Lie algebra
\cite{{Kirillov2004},{Kostant1970},{Woodhouse1991},{Souriau1997}}. 
This theory shares the viewpoint underlying the quantization considered in this paper.
The quantization of Lie-Poisson algebras was studied by Rieffel within
the framework of deformation quantization \cite{rieffel_90}.
As a related development, the deformation quantization of polynomial
Poisson algebras by means of universal enveloping algebras, generalizing
the corresponding construction for Lie-Poisson structures, was discussed
in \cite{Penkava_Vanhaecke}.

In \cite{Sako:2022pid,sako2024}, the author introduced a category that encompasses many quantization schemes and studied matrix regularization of
Lie-Poisson algebras as an example of how the inverse problem of
quantization can be treated in terms of categorical limits. 
(A categorical formulation of the classical limit of matrix regularization
is also discussed in \cite{Gohara:2019kkd,Jumpei:2020ngc}, where a different category is employed.)
Matrix regularization was introduced for algebraic varieties defined by polynomials in the Poisson center, 
and this theory was further developed in \cite{Gohara_Sako}.
The classical limit, however, was not discussed in \cite{Gohara_Sako}. 
In the present paper, we give a construction for algebraic varieties defined by
polynomials invariant under compact semisimple Lie algebras such that the
original algebraic variety is recovered in the classical limit.

\subsection{Purpose of this paper}

We study the quantization of spaces whose symmetries are described by a Lie-Poisson algebra. We focus in particular on algebraic varieties defined by polynomials invariant with respect to the Lie-Poisson bracket associated with compact semisimple Lie groups.
As a quantization of such spaces, we employ matrix regularization of Lie-Poisson algebras, which has recently been developed. In particular, using a matrix regularization based on reducible representations, we quantize algebraic varieties by means of their decomposition into coadjoint orbits.
\\

Recently, progress has been made by J. Gohara and the author on matrix regularization of Lie-Poisson algebras\cite{Gohara_Sako}. A matrix regularization, or weak matrix regularization, was formulated for algebraic varieties defined by fixing some of the Casimir polynomials,
$C^{k_i}(x)-\lambda_{k_i}=0$,
rather than by fixing all independent Casimir polynomials simultaneously.
In this formulation, matrix regularization was realized by mapping the remainder classes, computed with respect to the Gr\"obner basis of the ideal generated by $C^{k_i}(x)-\lambda_{k_i}$, to matrices through a representation.
However, other types of quantization, such as geometric quantization using irreducible representations, can be understood as the quantization of coadjoint orbits\cite{{Kirillov2004},{Kostant1970},{Woodhouse1991},{Souriau1997}}.
In fact, in the matrix regularization described above, the existence of elements in the kernel of the quantization originating from invariant polynomials results in a lack of faithfulness, even when irreducible representations are employed.
So it was argued that reducible representations are necessary in order to make the quantization more faithful. Although ideas for the quantization of ${\mathbb R}^3$ and $S^7$ were described, a detailed discussion was not provided in \cite{Gohara_Sako}.
\\

For a fixed element $\zeta \in \mathfrak{g}^*$, the coadjoint orbit through $\zeta$ is defined by
\begin{align}
O(\zeta) := \{ Ad^*_{g^{-1}}(\zeta) ~ | ~ g \in G \}.
\end{align}
If $G$ is compact, then $O(\zeta)$ is a closed embedded submanifold. 
In general, a Lie algebra equipped with the adjoint action of a Lie group is decomposed into the disjoint union of its adjoint orbits. For example, under the identification $su(2)\simeq {\mathbb R}^3$, the adjoint orbits are the spheres $S^2$ of fixed radius together with the origin. Thus ${\mathbb R}^3$ is exhausted by these adjoint orbits. 
This observation can be refined by using invariant polynomials.



Let $G$ be a compact connected semisimple Lie group with Lie algebra $\mathfrak{g}$. 
Using an invariant inner product, we identify $\mathfrak{g}$ with $\mathfrak{g}^*$. 
Let $\mathfrak{h}$ be a Cartan subalgebra of $\mathfrak{g}$, 
and let $\mathfrak{h}_+$ be a fixed closed dominant Weyl chamber. 
Every adjoint orbit intersects $\mathfrak{h}$ in a single Weyl-group orbit,
and hence intersects $\mathfrak{h}_+$ in exactly one point.
Moreover, the values of a set of homogeneous algebraically independent generators of the invariant
polynomial algebra determine each adjoint orbit uniquely  \cite[Theorem~4.37(ii)]{Alexandrino_Bettiol}.

Therefore, a common level set obtained by fixing all independent Casimir
polynomials is a single adjoint orbit, or equivalently a single
coadjoint orbit. If only some of the Casimir polynomials are fixed, the
resulting level set is decomposed into coadjoint orbits distinguished by
the values of the remaining Casimir polynomials.

For instance, in the case of $su(3)$, the level set defined by fixing the quadratic Casimir is identified, up to normalization, with $S^7\subset {\mathbb R}^8$. 
Since only the quadratic Casimir is fixed, this level set is not a single fiber of the full adjoint quotient map. 
Rather, it is decomposed into adjoint, or equivalently coadjoint, orbits, which are further distinguished by the values of the remaining invariant polynomials.

Based on these facts, in this paper we study the quantization of algebraic varieties defined by invariant polynomials associated with compact semisimple Lie groups. Such an algebraic variety is a subvariety of the Lie algebra and is decomposed into coadjoint orbits, which are submanifolds of the variety. In the matrix regularization considered here, we choose a family of coadjoint orbits that becomes dense in the algebraic variety in a suitable limit. 
Since irreducible representations provide more faithful quantizations of individual coadjoint orbits, the quantization of their disjoint union requires the use of reducible representations (see Appendix~\ref{App_red_rep}).
We then construct a reducible representation from the representations corresponding to these coadjoint orbits, and use it to define a matrix regularization.\\

Thus, in this paper, we construct matrix regularizations of algebraic varieties defined by polynomials in the Poisson center of the Lie-Poisson algebra associated with a compact semisimple Lie algebra.
As a concrete example, we study in detail the construction of fuzzy $S^7$.

The remainder of this paper is organized as follows.
In Section \ref{sect2}, we review the preliminaries and notation we use in this paper.
In Section \ref{sect3}, we show that the classical limit of a matrix regularization based on an appropriately chosen sequence of irreducible representations is a coadjoint orbit. More precisely, for a weak matrix regularization of an algebraic variety defined by fixing the value of a single Casimir invariant, we show that, by appropriately choosing the sequence of irreducible representations, the classical limit can be any coadjoint orbit contained in the algebraic variety. 
In Section \ref{sect4}, as a simple example illustrating the strategy of this paper, we construct fuzzy ${\mathbb R}^3$ as a weak matrix regularization densely filled with fuzzy spheres. In Section \ref{sect5}, as a more complicated example, we explicitly construct fuzzy $S^7$.
In Section \ref{sect6}, we show that, for any algebraic variety defined by fixing one Casimir polynomial of a compact semisimple Lie algebra, one can construct a weak matrix regularization based on reducible representations whose classical limit is the algebraic variety itself.
Section \ref{sect7} presents a summary and discusses directions for future work.


\section{Review and notation}\label{sect2}
Matrix regularization for Lie-Poisson algebras was introduced in \cite{Sako:2022pid}. Furthermore, in \cite{Gohara_Sako}, the quantization of algebraic varieties defined by invariant polynomials as algebraic equations was realized. Since we use these results in this paper, we review them in this section.
The precise definition of weak matrix regularization is given in \cite{Gohara_Sako}, and we refer the reader to that paper for the details. Roughly speaking, it is a sequence of quantizations into matrix algebras. More precisely, it is a sequence of linear maps from a Poisson algebra to matrix algebras such that the commutator of the images agrees with $\hbar$ times the image of the Poisson bracket, up to higher-order terms in $\hbar$. Such a sequence of quantizations is called a weak matrix regularization.
The conditions for weak matrix regularization are weaker than those required in the usual definition of matrix regularization. In particular, the usual definition often requires the classical space to be a symplectic manifold, whereas weak matrix regularization imposes no such requirement.
Indeed, in this paper, we consider the quantization of odd-dimensional spaces such as $\mathbb{R}^3$ and $S^7$.

\subsection{Weak matrix regularization of $A_{\mathfrak{g}}$}\label{sec2_1}
Let $\mathfrak{g}$ be a $d$-dimensional Lie algebra.
Let $e=\{  e_1 , e_2 , \cdots ,e_d \}$ be a fixed basis of $\mathfrak{g}$
satisfying commutation relations $[ e_i , e_j ]= f_{ij}^k e_k $,
where $f_{ij}^k$ are structure constants of $\mathfrak{g}$. 
For this Lie algebra $\mathfrak{g}$ we introduce a sequence of irreducible representations 
$\rho^{\mu} : \mathfrak{g} \rightarrow gl(V^\mu ) ( \mu =1,2, \cdots )$ and
a sequence $\hbar(\mu) ( \mu =1,2, \cdots )$ with $\hbar({\mu} ) \neq 0$.
(Here, for simplicity, we have taken $\mu=1,2,\dots$; however, $\mu$ is merely an index distinguishing the representations. Later, it will denote, for example, a highest weight, and should not be understood as necessarily taking values in the natural numbers.)
Here each $V^\mu$ is a finite dimensional vector space chosen as appropriate, and we put a condition
$\displaystyle \lim_{\mu \rightarrow \infty} \dim V^\mu = \infty$.
We denote the corresponding rescaled basis of $e$ by
\begin{align}\label{basis_matrix}
e^{(\mu)} =\{ 
\hbar(\mu) \rho^{\mu} (e_1 ), \hbar(\mu) \rho^{\mu} (e_2 ), 
\cdots &, \hbar(\mu) \rho^{\mu} (e_d ) \}
= \{ 
e_1^{(\mu)} , e_2^{(\mu)} , \cdots  ,e_d^{(\mu)} \} \\
[ e_i^{(\mu)} , e_j^{(\mu)} ]=& \hbar( \mu )f_{ij}^k e_k^{(\mu)} . \label{hbarCommRel}
\end{align}

Next, we introduce a Poisson algebra corresponding to this Lie algebra.
A well-known way to endow an algebra with a Poisson algebra structure
is to use the Kirillov-Kostant Poisson bracket.
(See for example 
\cite{matrix1,Kostant,Weinstein_Lu
}.)
Let $x=(x_1, x_2, \cdots , x_d) \in {\mathbb R}^d$ be commuting coordinates.
We consider a Lie-Poisson algebra 
of this coordinate polynomial ring $(\mathbb{C}[x]  , \cdot , \{ ~ , ~ \})$ by
$
\{ x_i , ~ x_j\}:= f_{ij}^k x_k , 
$
where $\displaystyle  f_{ij}^k \in {\mathbb C}$ are structure constants.
Concretely, this Poisson bracket is realized by
\begin{align*}
\{ f , ~ g\}:= f \omega g := f \overleftarrow{\partial}_i \omega_{ij} 
\overrightarrow{\partial}_j g := ({\partial}_i f) \omega_{ij} 
({\partial}_j g) 
\end{align*}
where $\displaystyle \partial_i = \frac{\partial}{\partial x_i}$ and
$\omega_{ij} = f_{ij}^k x_k $.
We denote the Poisson algebra $(\mathbb{C}[x]  , \cdot , \{ ~ , ~ \})$ by 
 $A_\mathfrak{g}$.
 Throughout this paper, the geometric level sets are regarded as real
algebraic subsets of $\mathbb{R}^d$, whereas their polynomial algebras
are taken over $\mathbb{C}$.
 We define degree of a monomial 
 $x^\alpha = (x_{1})^{\alpha_1} (x_{2})^{\alpha_2} \cdots  (x_{d})^{\alpha_d} $ by
 $\deg  x^\alpha :=|\alpha|:= \sum_{i=1}^d \alpha_i$. For the polynomial $f(x) = \sum_\alpha a_\alpha x^\alpha $, where
 $a_\alpha \in \mathbb{C}$, $\deg f(x) $ is defined by $\displaystyle \max_{a_\alpha \neq 0} \{  \deg x^\alpha \} $.
For multi-index, the notation $x^I =x_{i_1} x_{i_2} \cdots x_{i_m}$ 
is also often used below, where $\deg x^{I} = m=: |I|$.
 \\
 \bigskip

Next, let us construct quantization maps 
from the Lie-Poisson algebra $A_\mathfrak{g}$ to 
$T_{\mu}$ that is $\langle e^{(\mu)} \rangle$, i.e.
the $R$-algebra generated by $ e^{(\mu)} $. Here, $R$ denotes either $\mathbb{C}$ or $\mathbb{C}[\hbar]$, as appropriate.
(If $\mathfrak{g}$ is a semisimple Lie algebra, then $T_\mu$ is given by $End (V^\mu )= gl (V^\mu )$.) 
We choose a basis of $T_{\mu}$,
$E_1, E_2, \dots, E_D$, as polynomials of $e^{(\mu)}$.
Any polynomial of $e^{(\mu)}$ can be rewritten by $\hbar$ polynomial in $\langle \rho^\mu (e ) \rangle [\hbar ]$. 
So, a degree ${\rm deg}$ of any polynomial of $e^{(\mu)}$ can be defined by $\hbar$'s degree. 
Using $ E^i_{j_1 , \cdots , j_k}  \in {\mathbb C}$, $E_i~ (i=1, \cdots , D)$ are expressed as
$$ E_i  = \sum_k E^i_{j_1 , \cdots , j_k} e^{(\mu)}_{j_1} \cdots  e^{(\mu)}_{j_k}
= \sum_k \hbar^k(\mu)  E^i_{j_1 , \cdots , j_k} \rho^\mu (e_{j_1}) \cdots  \rho^\mu (e_{j_k}),
$$
where $E^i_{j_1 , \cdots , j_k}$ is independent of $\hbar$.
Note that Einstein summation convention is used for each $j_l$.
For each $E_i$, such expression given by $e^{(\mu)}$ is not unique in general.
We choose an expression that minimizes $\displaystyle \max_{E^i_{j_1 , \cdots , j_k}\!\!\!\!\!\!\!\!\!\!  \neq 0 } \{k \}$.
Then there exists the degree of $\hbar$ of $E_i$ i.e., 
$\displaystyle \deg E_i := \max_{E^i_{j_1 , \cdots , j_k}\!\!\!\!\!\!\!\!\!\!  \neq 0 } \{k \}$.
The highest degree of $\{E_1, E_2, \dots, E_D \}$ is denoted by $n_\mu$, i.e.,
$n_\mu  = \max \{ {\deg}E_1, \cdots , {\rm deg}E_D \}$.
$n_\mu$ does not depend on the choice of $\{E_1, E_2, \dots, E_D \}$.

Any $M(\hbar(\mu) ) \in gl (V^\mu ) [\hbar(\mu)] $ is expressed as
$$M(\hbar(\mu) )= \sum_{0 \le k} \hbar(\mu)^k M_k = \sum_{0 \le k \le n_\mu} \hbar(\mu)^k M_k + \tilde{O}(\hbar(\mu)^{n_\mu +1}), $$
where each $M_k \in gl(V^\mu)$ does not depend on $\hbar(\mu)$. 
(The precise definition of the notation $\tilde{O}(\hbar(\mu)^n)$ is given in the Appendix of \cite{Gohara_Sako}. Roughly speaking, it means that, when regarded as a polynomial in the variable $\hbar(\mu)$, every term has degree at least $n$.
)
For any $M(\hbar(\mu) )$, we define the projection map
$R_\mu $ by
\begin{align}
R_\mu (M) := \sum_{0 \le k \le n_\mu} \hbar(\mu)^k M_k .
\end{align}



\begin{definition}\label{def_FuzzyQ} 
We define a linear map
$\tilde{q}_\mu : A_\mathfrak{g} \rightarrow T_\mu $
by
\begin{align} \label{tildeq_mu}
\sum_I f_I x^I := \sum_k f_{i_1, \cdots , i_k} x_{i_1} \cdots x_{ i_k} \mapsto 
\sum_I f_I e^{(\mu)}_{(I)} = \sum_k f_{i_1, \cdots , i_k} e^{(\mu)}_{(i_1, \cdots , i_k)} 
\end{align}
where $f_{i_1, \cdots , i_k} \in \mathbb{C}$ is completely symmetric
in the indices $i_1,\ldots,i_k$,
\begin{align*}
e^{(\mu)}_{(I)}:= e^{(\mu)}_{(i_1, \cdots ,i_k)}:=
\frac{1}{k!}\sum_{\sigma \in Sym(k)}
e^{(\mu)}_{i_{\sigma(1)}} \cdots e^{(\mu)}_{ i_{\sigma(k)}} ,
\end{align*}
and we require 
$\tilde{q}_\mu (1) = Id \in gl(V^\mu)$.
We also define $q_\mu : A_\mathfrak{g} \rightarrow T_\mu $
by $R_\mu \circ \tilde{q}_\mu $ , i.e.
\begin{align} \label{q_mu}
\sum_I f_I x^I := \sum_k f_{i_1, \cdots , i_k} x_{i_1} \cdots x_{ i_k} \mapsto 
\sum_I f_I e^{(\mu)}_{(I)} = \sum_k^{n_\mu} f_{i_1, \cdots , i_k} e^{(\mu)}_{(i_1, \cdots , i_k)}  .
\end{align}
\end{definition}
For any $f,g \in A_{\mathfrak{g}}$ with $\deg f + \deg g \le n_\mu$, there exists $P = \sum_i^D c_i(\hbar(\mu)) E_i  \in T_\mu$
with $c_i(\hbar) \in {\mathbb C}[\hbar (\mu )]$ satisfying
\begin{align}
\label{asym_hom_mu}
 q_{\mu} ( f )  q_{\mu} ( g ) = q_{\mu} ( f  g ) +\hbar(\mu) P. 
\end{align}





Any semisimple Lie algebra 
admits a sequence of irreducible representations 
$\{ V^\mu \}$ that satisfies the following condition.
The Casimir operators $C_{i}$ act on $V^\mu$ as scalar multiples of the identity:
\begin{align}
C_{i} = C_i(\mu) Id_\mu , \quad
\lim_{\dim V^\mu \to \infty} |C_i(\mu) | = \infty .
\label{okubo} 
\end{align}
Here $Id_\mu$ denotes the identity operator on $V^\mu$, and
$C_i(\mu) $ is the eigenvalue of $C_{i}$ in the representation $V^\mu$.

Casimir polynomials of $A_\mathfrak{g}$ are defined by
$
\{ x_i , f(x) \} = 0~ ( i = 1, \cdots , d ) ,
$
and we denote the set of all  Casimir polynomials of $A_\mathfrak{g}$ by $CaP$.
\begin{proposition}
For any Casimir polynomial 
$ f^C  \in CaP$, $ q_\mu  (f^C ) \in T_\mu$ 
 is a Casimir operator.
\end{proposition}
Let $C^k(x)$ be a Casimir polynomial given by a homogeneous polynomial of degree $d_k$.
By (\ref{okubo}), 
$\lim_{\dim V^\mu \rightarrow \infty} |C_i(\mu) | = \infty$. 
We will uniquely determine $\hbar(\mu)$ as follows.
Let $C^k(x) = \sum_\alpha C_{k, J} x^{J}= \sum_\alpha C_{k , J} x_{j_1}\cdots x_{j_{d_k}} $ be a linearly independent  $d_k$-th-degree Casimir polynomial.
We define $C_k( e^\mu) $ by
\begin{align}
C_k( e^\mu) := q_\mu ( C^k(x) )
= \sum_J C_{k ,J}  e^{(\mu)}_{(j_{1} , \cdots , j_{d_k})} .
\end{align}
Under  (\ref{okubo}), it can be fixed to any one eigenvalue $\lambda_k \in {\mathbb C}$ 
of the matrix $C_k( e^\mu)$
by determining the sequence of $\hbar(\mu)$ 
and $V^\mu$ appropriately;
\begin{align} \label{fixed_casimir_relation}
C_k( e^\mu) = \hbar^{d_k} (\mu ) C_{k}
= \lambda_k Id_\mu
\end{align}
for any $q_\mu : A_\mathfrak{g} \to \operatorname{End}(V^\mu) $.
Here, we take $\hbar(\mu)$ to be real and, whenever there is a sign ambiguity, choose the positive value.

\subsection{(Weak) Matrix regularization for Lie-Poisson varieties $A_\mathfrak{g} /I(C) $}\label{sect2_2}
Consider a vector space $C \subset CaP$ whose basis is $\{ f_i^C \} $, i.e.,
$C:= \{ \sum_i a_i f_i^C ~;~ a_i \in  \mathbb{C},  ~f_i^C \in CaP \}$.
We introduce an ideal of $A_\mathfrak{g} $ generated by $C$ as $I(C)$.
This ideal is compatible with the Poisson structure because $\{ x_i , f_j^C(x) \}=0$.
So, we can introduce the new Poisson algebra as follows:
\begin{align}
A_\mathfrak{g}  / I(C) := \left\{ [f(x)]  ~ | ~ f(x) \in A_\mathfrak{g} 
\right\} ,
\end{align}
where $ [f(x)] = \{ f(x) + h(x) ~|~ h(x) \in I(C) \} $, and the sum and multiplication are defined
as $[f(x)]+[g(x)]= [f(x)+g(x)]$ and $[f(x)]\cdot [g(x)] = [f(x)\cdot g(x)]$.
The Poisson bracket is also defined by 
as $ \{ [f(x)] , [g(x)] \} := [ \{ f(x) , g(x) \}] .$
We abbreviate this Poisson algebra $(A_\mathfrak{g}  / I(C) , \cdot  , \{ ~,~ \})$ as $A_\mathfrak{g}  / I(C)$.

The quotient and remainder of a multivariate polynomial cannot be uniquely determined in general. However, after choosing a monomial ordering, such as the lexicographic order, and thereby inducing an ordering on multivariate polynomials, the remainder of a polynomial can be uniquely defined.

It is known that if we fix the ordering every ideal $I$ has a unique reduced Gr\"obner basis,
and  the following fact is known. 

\begin{theorem}\label{reducedGrobner} {\rm (See for example \cite{Dumniit_Foote_Abstract Algebra
}.)}
Fix a monomial ordering on ${\mathbb C} [x]$.
1) Every ideal $I \subset  {\mathbb C} [x]$ has a unique reduced Gr\"obner basis.
2)  Let ${g_1, \dots , g_m}$ be the Gr\"obner 
basis for the ideal $I$ in $ {\mathbb C} [x]$. 
Then every polynomial $f \in  {\mathbb C} [x]$ can be written uniquely in the form 
\begin{align}
f= h_f + r_f \label{IplusG}
\end{align}
where $h_f  \in  I$ and no monomial term of the $r_f$ is divisible by any leading term of $g_i $. 
\end{theorem}
In the following, we fix a monomial ordering by the graded lexicographic ordering.


The Gr\"{o}bner basis becomes important only in special cases where the Lie algebra and the sequence of representations are chosen so that the eigenvalues of several Casimir operators can be fixed simultaneously by a single parameter $\hbar(\mu)$.
For example, let the Lie algebra be of the form $su(n)\oplus \cdots \oplus su(n)$. If one takes a sequence of representations obtained by using the same sequence of representations for each $su(n)$ factor, then one can construct examples described by a system of equations $C^k_i(x) =\lambda^k_i (i=1, \dots )$.
In the present paper, however, we do not treat such cases. Instead, we restrict ourselves to cases in which the Gr\"{o}bner basis property does not play an essential role.



Let us choose the generators of the ideal $I(C) \subset A_\mathfrak{g}$ for some fixed $k$ as
$ 
( f^C(x):= C^k (x) - \lambda_k ) 
$
i.e.,
\begin{align}
I(C)& := \left\{  f^C(x) g(x) \in  A_\mathfrak{g} ~ | ~ 
f^C(x) =  C^k (x) - \lambda_k , ~
g(x) \in  A_\mathfrak{g} ~ \right\} .  \label{2_13}
\end{align}


Next, we introduce $q_{A/I,\mu}$, which is a quantization of the space defined by $C^k (x) = \lambda_k$.
Although this differs from the construction of $q_{A/I,\mu}$ given in \cite{Gohara_Sako}, the following provides the shortest way to define the same object \cite{Gohara_phd}.
For the Gr\"obner basis $G$, the linear map ${\mathcal R}_G : A_\mathfrak{g} / I(C) \rightarrow A_\mathfrak{g}$
\begin{align}
[f] \mapsto r_{f}
\end{align}
is well-defined, where $r_{f}$ is the remainder 
by dividing $f$ by the Gr\"obner basis $G$.
\begin{definition}
When $n_\mu$ is larger than $d_k$,  
 $q_{A/I, \mu} : A_\mathfrak{g} / I(C) \rightarrow \operatorname{End}(V^\mu )$ is  defined by
\begin{align}
q_{A/I, \mu} = R_\mu \circ \tilde{q}_\mu \circ {\mathcal R}_G = q_\mu \circ {\mathcal R}_G.
\label{q_AImu_general}
\end{align}
\end{definition}
For $f= h_f + r_{f} $ with
$ \displaystyle
r_{f}= \sum_{J} a_{J} x^{J}
= \sum_{\deg x^{J} \le n_\mu } \!\!\!\!\!
a_{J} x^{J}
+
 \sum_{\deg x^{J} > n_\mu  } \!\!\!\!\!
a_{J} x^{J} 
$, the explicit calculation of 
$q_{A/I, \mu}  : A_\mathfrak{g} / I(C)\to \operatorname{End}(V^\mu)$ is given as
\begin{align}
q_{A/I, \mu}  ([f(x)] ) 
= \sum_{m \le n_\mu } \!
a_{J}~ e_{(j_1 , \cdots  , j_m )}^{(\mu )}.
\end{align}
As can be easily seen from the definition of $q_{A/I, \mu}$, when we chose $I(C) =\{ 0 \} $, 
it is the same as $q_\mu$. 
Therefore, $q_{A/I, \mu}$ is a generalization of $q_\mu$ .\\

For this $q_{A/I ,\mu}$,
the following theorems are obtained.
\begin{theorem}\label{prop4_12}
For any $[f], [g] \in A_\mathfrak{g} / I(C)$ with $\deg r_f +\deg r_g \le n_\mu$ and $\deg r_{fg} \le n_\mu$,
there exists $P = \sum_i^D c_i(\hbar(\mu)) E_i  \in T_\mu$
with $c_i(\hbar) \in {\mathbb C}[\hbar (\mu )]$ satisfying
$$
q_{A/I, \mu} ( [ f ] )  q_{A/I, \mu} ( [ g ] ) = q_{A/I, \mu} ( [f] [g]) + \hbar(\mu)  P .
$$
Here $r_f$, $r_g$, $h_f$, and $h_g$ are used in the sense defined in Theorem \ref{reducedGrobner}.
\end{theorem}

\begin{theorem}
$q_{A/I, \mu} $
is a weak matrix regularization, i.e., for $\forall f, g \in A_\mathfrak{g} $, 
there exists $P = \sum_i^D c_i(\hbar(\mu)) E_i  \in T_\mu$,
where each $c_i(\hbar(\mu)) $ is a polynomial in $\hbar(\mu)$, such that
\begin{align}
[ q_{A/I, \mu} ( [ f ] ) ,  q_{A/I, \mu} ( [ g ] ) ] = \hbar(\mu ) q_{A/I, \mu} ( \{ [f] , [g] \} ) + \hbar^2(\mu) P.
\label{q_AI_quantization}
\end{align}
\end{theorem}


\begin{rem}
The map introduced above,
$
q_{A/I,\mu}: A_\mathfrak{g}/I(C)\to \operatorname{End}(V^\mu),
$
was defined in \cite{Gohara_Sako} as follows. It is easy to see that this definition is equivalent to the one given above \cite{Gohara_phd}.
\\
First, we introduce a quantization map $q_U$ from $A_\mathfrak{g} $ 
 to  enveloping algebra $ \mathcal{U}_\mathfrak{g} $
 by $q_U (x_{\alpha_1} \cdots x_{\alpha_m}) = X_{(\alpha_1, \cdots , \alpha_m )}$.
 Here $X_1,\ldots,X_d$ are the generators of the enveloping algebra $U(\mathfrak g)[\hbar]$. 
Next, we construct a quantization map 
$q_{U/I}  : A_\mathfrak{g}  /I(C) \to  \mathcal{U}_\mathfrak{g} [\hbar]/I(C(X))$
for nontrivial $I(C)$. 
$I(C)$ is not arbitrary, but (\ref{2_13}).
We use $q_\mu$ to obtain the relation (\ref{fixed_casimir_relation}).
Let $G$ be the reduced Gr\"obner basis of $I(C)$.
For any $f(x) \in  {\mathbb C}[x]$ $f(x)= r_f (x) + h_f(x)$ is uniquely determined by $G$, 
where $h_f(x) \in I(C)$ and $r_f(x) \notin I(C)$.
Then we can define $q_{U/I}$ by $q_{U/I}  ([f(x)]) :=  [ q_{U}  ( r_{f, G}(x) ) ]$.
There exists an algebra homomorphism 
$\rho_{U/I , \mu} :  \mathcal{U}_\mathfrak{g}[\hbar] / I(C(X))  \to \operatorname{End}(V^\mu )$.
Using this $\rho_{U/I , \mu}$ and a projection operator
$R_\mu : \operatorname{End} (V^\mu )[\hbar(\mu)]  \to \operatorname{End} (V^\mu ) [\hbar(\mu)] $ 
that restricts the degree of $\hbar$ to $n_\mu$ or less. 
Finally, we get the weak matrix regularization $A_\mathfrak{g}  /I(C) \to  \operatorname{End} (V^\mu )$ by
$q_{A/I ,\mu} : = R_\mu \circ  \rho_{U/I , \mu} \circ q_{U/I} $.
\begin{align*}
\vcenter{
\xymatrix@C=10pt@R=4pt{
 A_\mathfrak{g}  / I(C)  \ar[r]^{ q_{U/I} \in Q}  \ar@/^20pt/[rr]^{ q_{\mu}^{pre} \in Q}
 \ar@/^40pt/[rrr]^{ q_{A/I , \mu} \in Q}& 
 \mathcal{U}_\mathfrak{g}[\hbar ] / I(C(X)) \ar[r]^{\rho_{U/I, \mu}} &
  \operatorname{End}(V^\mu) \ar[r]^{R_\mu} &  \operatorname{End}(V^\mu)\\
   {}&{}&{}&{}\\
\big[f(x)\big]=  \big[r_{f} + h_f \big]   \ar@{|->}[r]& 
[ q_U(r_f ) ]= \big[\sum_I a_I X_{(i_1, \cdots , i_m )} \big]  \ar@{|->}[r] &  
\sum_I a_I e^{(\mu)}_{(i_1 , \cdots , i_m)}
 \ar@{|->}[r] & \!\!\!\!\!
\sum_{|I| = m< n_\mu} \!\!\! a_I e^{(\mu)}_{(i_1 , \cdots , i_m)} 
}}
\end{align*}
This was the original definition of the map $q_{A/I , \mu}$.
\end{rem}

\subsection{Matrix regularization using reducible representations}

\label{rev_sect6.1}
We consider the following decomposition of a reducible representation into a direct sum of irreducible representations of a semisimple Lie algebra $\mathfrak{g}$;
\begin{align*}
V^\mu &=\bigoplus_{a=1}^{m_\mu} V_\mu ^a.
\end{align*}
We define a projection $\hat{P}_a:V^\mu\to V_\mu^a$, and a reducible representation $\rho^\mu$ of $\mathfrak{g}$
\begin{align*}
\rho^\mu &= \sum_a^{m_\mu} \rho_{\mu ,a}\hat{P}_a,\\
\rho_{\mu,a}&:\mathfrak{g}\to End(V_\mu ^a).
\end{align*}
Suppose that each representation satisfies $[\rho_{\mu,a}(e_i),\rho_{\mu,a}(e_j)]=f^k_{ij}\rho_{\mu,a}(e_k)$. Then
\begin{align*}
[\rho ^\mu(e_i),\rho^\mu(e_j)]&=\sum_{a,b}[\rho_{\mu,a}(e_i)\hat{P}_a,\rho_{\mu,b}(e_j)\hat{P}_b]
=f^k_{ij}\rho^\mu(e_k).
\end{align*}
Let us assume that $\hbar(\mu)$ is given for each representation
\begin{align*}
\hat{\hbar}(\mu) = \hbar_1(\mu )Id_1\oplus \cdots \oplus \hbar_{m_\mu}(\mu)Id_{m_\mu} = \hbar(\mu ) 
( {\mathrm r}_1(\mu )Id_1\oplus \cdots \oplus {\mathrm r}_{m_\mu}(\mu)Id_{m_\mu} ) 
,
\end{align*}
where $Id_a:= Id_{V_{\mu}^a}~ (a= 1, \dots , m_\mu)$, $\hbar_a(\mu) \in {\mathbb R}~ (a= 1, \dots , m_\mu)$, 
and ${\mathrm r}_a(\mu) \in {\mathbb R}~ (a= 1, \dots , m_\mu)$
satisfies $\hbar(\mu) {\mathrm r}_a(\mu) = \hbar_a(\mu) $. 
Here, the direct sum of algebras is represented 
by block-diagonal matrices. 
Following this convention, we shall also describe direct products of non-matrix algebras using the direct sum notation in what follows.
We also define each basis as follows
\begin{align*}
e^{(\mu)}_i&:=\hat{\hbar}(\mu)\rho^\mu(e_i)=\hbar_1(\mu)\rho_{\mu,1}(e_i)\oplus \cdots \oplus \hbar_{m_\mu}(\mu)\rho_{\mu,m_\mu}(e_i),\\
e^{(\mu,a)}_i&:=\hbar_a(\mu)\rho_{\mu,a}(e_i).
\end{align*}
From the commutation relations of the representation $\rho_{\mu,a}$,
\begin{align*}
[e^{(\mu,a)}_i,e^{(\mu,a)}_j]&=\hbar_a(\mu)f^k_{ij}e^{(\mu,a)}_k,\\
[e_i^{(\mu)},e_j^{(\mu)}]&=\hat{\hbar}(\mu)f^k_{ij}e^{(\mu)}_k.
\end{align*}
For $T_\mu^a:=\langle e^{(\mu,a)}\rangle$, we redefine $T_\mu:=T^1_\mu\oplus \cdots \oplus T^{m_\mu}_\mu$ as an $R$-algebra generated by $e^{(\mu)}$. Given the basis $E_i^a$ $(i=1,\cdots,D_a)$ of $T_\mu^a$, we set $n_\mu^a:=\max(\deg E_1^a,\cdots, \deg E_{D_a}^a)$, and proceed to redefine the quantization map for the reducible representations.
\begin{definition}\label{def6_1}
We define a linear map $q_\mu^a:A_{\mathfrak{g}}\to T_\mu^a$ by
\begin{align*}
q^a_\mu \Big(\sum_I f_Ix^I \Big)=\sum_{|I|\le n^a_\mu}f_I e^{(\mu,a)}_{(I)},
\end{align*}
and a linear map $q^R_\mu:A_\mathfrak{g}\to T_\mu$ by
\begin{align*}
q^R_\mu \Big( \sum_I f_Ix^I \Big)=\sum_{a=1}^{m_\mu}\sum_{|I|\le n_\mu^a}f_Ie_{(I)}^{(\mu,a)}\hat{P}_a=\sum_{a=1}^{m_\mu} q^a_\mu \Big(\sum_If_Ix^I \Big)\hat{P}_a,
\end{align*}
i.e. $q^R_\mu:=q^1_\mu\oplus \cdots \oplus q^{m_\mu}_\mu$.
\end{definition}
Under these settings, results analogous to similar propositions and theorems stated in subsection \ref{sect2_2} can be derived as follows.

\begin{theorem}
Let $\{ q_\mu^R:A_\mathfrak{g}\to T_\mu \}$ be a sequence of quantizations defined by Definition \ref{def6_1}. 
Suppose that $\{\hbar(\mu)\}$ is a sequence such that $\max_{1\le a \le m_\mu} |\hbar_a(\mu)| \to 0 $ as $\dim V_{\mu}^a \to \infty$. Then $q_\mu^R$ is a weak matrix regularization. In other words,
\begin{align*}
[q^R_\mu(f),q^R_\mu(g)]=\hat{\hbar}(\mu)q^R_\mu(\{f,g\})+\hat{\hbar}^2(\mu)P,
\end{align*}
where $P=\sum_{a=1}^{m_\mu}\sum_la_{i_1\cdots i_l}^{\mu,a}q^a_\mu(x_{i_1})\cdots q^a_{\mu}(x_{i_l})\hat{P}_a$.
\end{theorem}
\begin{proposition}
For any $f,g\in A_{\mathfrak{g}}$ with $\deg f+\deg g \le \min(n_\mu^1,\cdots,n_\mu^{m_\mu})$, there exists 
$\displaystyle P=\sum_{a=1}^{m_\mu}\sum_{i=1}^{D_a} c_i^a(\hbar_a(\mu))E^a_i\hat{P}_a \in T_\mu$ with $c_i^a(\hbar_a(\mu))\in \mathbb{C}[\hbar_a(\mu)]$ satisfying
\begin{align*}
q_\mu^R(f)q^R_\mu(g)=q^R_\mu(fg)+\hat{\hbar}(\mu)P.
\end{align*}
\end{proposition}
\bigskip

Let us introduce Casimir operators $C_{k}(V_\mu^a)Id_{\mu,a} \ (a= 1, \dots , m_\mu )$ of each $V_\mu^a$ such that $|C_k(V_\mu^a)|\to \infty $ when $\dim V_\mu^a\to \infty$. For a fixed $d_k$-degree homogeneous Casimir polynomial $C^k(x)=\sum_J C^k_{J}x^J$ in $A_\mathfrak{g}$, we defined $C^{a}_k(e^{(\mu,a)})$ by
\begin{align*}
C^{a}_k(e^{(\mu,a)})&:=q^a_\mu(C^k(x))\\
&=\frac{(\hbar_a(\mu))^{d_k}}{d_k!}\sum _J C_{J}^k \sum_{\sigma \in S_{d_k}}\rho_{\mu,a}(e_{j_{\sigma (1)}})\cdots \rho_{\mu,a}(e_{j_{\sigma (d_k)}})\\
&=(\hbar_a(\mu))^{d_k} C_{k}(V_\mu^a)Id_{\mu,a}.
\end{align*}

As in the case of irreducible representations, we fix the parameter 
$\lambda_k\in \mathbb{C}$ determining the ideal, and choose each $\hbar_a(\mu) $  so as to satisfy the following condition
\begin{align*}
q^a_\mu(C^k(x))=
(\hbar_a(\mu))^{d_k} C_{k}(V_\mu^a)Id_{\mu,a}=\lambda_k Id_{\mu,a},
\end{align*}
for any $q_\mu^a:A_\mathfrak{g}\to \operatorname{End} V^\mu_a (a= 1, \dots , m_\mu )$. 
$\lambda_k\in \mathbb{C}$ is the eigenvalue of the matrix $C_k^{a}(e^{(\mu,a)})$.
Then, we obtain
\begin{align*}
q^R_\mu(C^k(x))&=\sum_{a=1}^{m_\mu} q^a_\mu(C^k(x))\hat{P}_a
=\lambda_k Id_{\mu}.
\end{align*}


\begin{definition}\label{def_2_10}
We apply (\ref{q_AImu_general}) to each irreducible component $V_\mu^a$. Namely, we introduce
\[
q_{A/I,\mu}^a : A_\mathfrak{g}/I(C) \rightarrow gl(V_\mu^a)
\]
by
\begin{align}
q_{A/I,\mu}^a = q_\mu^a \circ {\mathcal R}_G .
\label{q_AImu_general2}
\end{align}
Using these maps, we define the quantization associated with the reducible representation as the block-diagonal map
 $q_{A/I, \mu}^R : A_\mathfrak{g} / I(C) \rightarrow \operatorname{End}(V^\mu),$ given by
 \begin{align}
q_{A/I, \mu}^R := q_{A/I, \mu}^1 \oplus \cdots \oplus q_{A/I, \mu}^{m_\mu} .
\end{align}
\end{definition}

Finally, we obtain the following.
\begin{theorem}
$q_{A/I, \mu}^R $
is a weak matrix regularization, i.e., for  any $[f], [g] \in A_\mathfrak{g} / I(C)$, 
there exists $\displaystyle P = \sum_{a=1}^{m_\mu} \sum_i^{D_a} c_i^a(\hbar_a(\mu)) E_i^a \hat{P}_a \in T_\mu$,
where each $c_i^a(\hbar_a(\mu)) $ is a polynomial in $\hbar_a(\mu)$, such that
\begin{align}
[ q_{A/I, \mu}^R ( [ f ] ) ,  q_{A/I, \mu}^R ( [ g ] ) ] = \hat{\hbar}(\mu ) q_{A/I, \mu}^R ( \{ [f] , [g] \} ) + (\hat{\hbar}(\mu))^2 P.
\label{q_AI_R_quantization}
\end{align}
Furthermore, 
if $\displaystyle \deg r_f +\deg r_g \le \min_{1\le a \le m_\mu} \{n_\mu^a\}$ and $\displaystyle \deg r_{fg} \le  \min_{1\le a \le m_\mu} \{n_\mu^a\}$,
$q_{A/I, \mu}^R $ also satisfies 
$$
q_{A/I, \mu}^R ( [ f ] )  q_{A/I, \mu}^R ( [ g ] ) = q_{A/I, \mu}^R ( [f] [g]) + \hat{\hbar}(\mu)  P 
$$
with some $\displaystyle P = \sum_{a=1}^{m_\mu} \sum_i^{D_a} c_i^a(\hbar_a(\mu)) E_i^a \hat{P}_a  \in T_\mu$.
\end{theorem}


\section{Irreducible representation and classical limit }\label{sect3}
The quantization maps $q_\mu$ and $q_{A/I,\mu}$ reviewed in Subsection \ref{sect2_2}, which are associated with irreducible representations, are generally not quantizations of coadjoint orbits.
This means that the source of $q_\mu$ or $q_{A/I,\mu}$
is not a coadjoint orbit, in general.
However, in the classical limit $\hbar\to 0$, which corresponds to the large-$N$ limit of the matrix model, matrix regularizations based on irreducible representations generally converge to coadjoint orbits in the sense that the coimages of the quantization maps approach the algebras of functions on the corresponding coadjoint orbits.
We now explain this point.

\subsection{Classical Limits of the Eigenvalues of Casimir Operators and Coadjoint Orbits}
Let $\mathfrak{g}$ be a $d$-dimensional Lie algebra of a compact and connected Lie group.
Thus, we consider the case where the Lie algebra is semisimple or reductive \cite{GHV}. In this case, the Casimir polynomials, which are invariant polynomials of the Poisson algebra $(A_{\mathfrak{g}},\{\ ,\ \})$, are described by $r=\operatorname{rank}\mathfrak{g}$ algebraically independent homogeneous polynomials.

We denote these Casimir polynomials by
\begin{align}
C^1(x),\ldots,C^r(x),
\end{align}
and choose $C^i(x)$ to be a homogeneous polynomial of degree $d_i$. We consider the case where the ideal is generated by a single polynomial $C^k(x)-\lambda_k$:
$$
I(C):=(C^k(x)-\lambda_k).
$$
Here $\lambda_k $ is a real number such that $\{ x~ |~ C^k(x)-\lambda_k=0\} \neq 0$. 
For an ideal generated by a single polynomial, the Gr\"{o}bner basis is the polynomial itself. In what follows, we restrict ourselves to such cases.

For each representation space $V^\mu$, the parameter $\hbar(\mu)$ is fixed by imposing, for a fixed $k$, the condition
\begin{align}
q_\mu(C^k(x))
= (\hbar(\mu))^{d_k} C_k(\mu)\operatorname{Id}_{V^\mu}
= \lambda_k \operatorname{Id}_{V^\mu}.
\end{align}
Here $C_k(\mu)$ is the eigenvalue of the Casimir operator determined by the representation space $V^\mu$. It differs from the eigenvalue of $q_\mu(C^k(x))$ by the factor $(\hbar(\mu))^{d_k}$. The label $\mu$ distinguishing the representation spaces may be taken, for example, to be the highest weight.
\begin{proposition}
\label{Syusoku1}
Let $q_{\mu_j} \ (j=1,2, \dots )$ denote the matrix regularization constructed by the method described in Subsection~\ref{sect2_2} from the following sequence of representations and $\hbar(\mu_j)$.
Let $\{\mu_j\}_{j=1}^{\infty}$ be a sequence of highest weights of representations of a semisimple Lie algebra, written as
$ \displaystyle
\mu_j=\sum_{i=1}^r m_i(j)\omega_i,
\
m_i(j)\in {\mathbb Z}_{\ge 0}.
$
Assume that
$\displaystyle
\lim_{j\to\infty}\|\mu_j\|=\infty
$
and that
\[
\lim_{j\to\infty}\frac{\mu_j}{\|\mu_j\|}
=\mu_\infty
=
\sum_{i=1}^r m_i\omega_i,
\qquad
m_i\in{\mathbb R}_{\ge 0}.
\]
Suppose that $\hbar(\mu_j)$ is chosen by fixing one Casimir polynomial $C^k(x)$ of degree $d_k$, namely
$\displaystyle
\bigl(\hbar(\mu_j)\bigr)^{d_k}C_k(\mu_j)=\lambda_k,
$
where $\lambda_k$ is a fixed nonzero real number. Assume moreover that the leading homogeneous part of $C_k(\mu_j)$ does not vanish in the limiting direction $\mu_\infty$. Then, for each $n=1,\ldots,r$, there exists a real number $\lambda_n$ such that
\begin{align}
\lambda_n
=
\lim_{j\to\infty}
\bigl(\hbar(\mu_j)\bigr)^{d_n}C_n(\mu_j).
\end{align}
\end{proposition}

\begin{proof}
By the Weyl dimension formula, if we denote the set of positive roots by $\Phi^+$, then the dimension of the representation space is given by
\begin{align*}
\dim V^\mu
=
\prod_{\alpha \in \Phi^+}
\frac{(\mu+\rho,\alpha)}{(\rho,\alpha)} .
\end{align*}
Here $\rho$ is the half-sum of the positive roots. Therefore, if we take a sequence of representations whose highest weights diverge to infinity, then the dimensions of the representation spaces
$V^{\mu_1},V^{\mu_2},V^{\mu_3},\ldots$
also diverge, namely
\begin{align*}
\lim_{j\to\infty}\dim V^{\mu_j}=\infty .
\end{align*}
Thus, this sequence of irreducible representations realizes a weak
matrix regularization.
If the limiting direction of $\mu_j$ is regular, then one has the asymptotic behavior
\begin{align*}
\dim V^{\mu_j}\sim \|\mu_j\|^{\sharp \Phi^+}.
\end{align*}
In other words, when all $m_i$ are rescaled as $m_i (j) \mapsto n m_i (j)$, the highest-degree term is multiplied by $n^{\sharp \Phi^+}$.

Let $\mathfrak h$ be a Cartan subalgebra of $\mathfrak g$. 
By the Harish-Chandra isomorphism, each central element
$z\in Z(U(\mathfrak g))$
corresponds to a Weyl-invariant polynomial $P_z$ on $\mathfrak h^*$. The corresponding central character on the irreducible highest-weight module $V^\mu$ is obtained by evaluating this polynomial at the shifted weight $\mu+\rho$:
\begin{align}\label{HC}
\chi_\mu(z)=P_z(\mu+\rho).
\end{align}
Here $\rho$ is the half-sum of the positive roots \cite{Humphreys1972}.
In the present case, the eigenvalue of the Casimir operator associated with $C^n(x)$ is given by
\begin{align}
C_n(\mu_j)
=
\chi_{\mu_j}\bigl(q_U(C^n(x))\bigr).
\end{align}
Equivalently, the operator $q_{\mu_j}(C^n(x))$ acts on $V^{\mu_j}$ by the scalar
$C_n(\mu_j)$, up to the scaling by the appropriate power of $\hbar(\mu_j)$.
Here we define the canonical linear map
\[
q_U:A_\mathfrak{g}\to U(\mathfrak g)[\hbar]
\]
by symmetrization:
\begin{align}
x_{i_1}\cdots x_{i_m}
\mapsto
\frac{1}{m!}
\sum_{\sigma\in {Sym}(m)}
X_{i_{\sigma(1)}}\cdots X_{i_{\sigma(m)}} ,
\notag
\end{align}
where $X_1,\dots,X_d$ are the generators of the enveloping algebra $U(\mathfrak g)[\hbar]$. For details, see \cite{Gohara_Sako}.

Since $C_n(\mu_j)$ is the eigenvalue associated with a homogeneous polynomial of degree $d_n$, under the rescaling of the basis
$
\rho^\mu(e_i)\mapsto e_i^{(\mu)}
=
\hbar(\mu)\rho^\mu(e_i),
$
the eigenvalue is rescaled as
$
C_n(\mu_j)
\mapsto
\bigl(\hbar(\mu_j)\bigr)^{d_n}C_n(\mu_j).
$
Let us denote
$\displaystyle
\lim_{j \rightarrow \infty}\frac{\mu_j+\rho}{\| \mu_j \|}
$ by $\mu_\infty$.
(\ref{HC}) gives
\begin{align}
\| \mu_j \|^{-d_n}C_n(\mu_j)
&=
\| \mu_j \|^{-d_n}P_{C_n}(\mu_j+\rho)
\rightarrow
P_{C_n}^{[d_n]}(\mu_\infty),
\label{scaled_Casimir_limit}
\end{align}
where $P_{C_n}^{[d_n]}$ denotes the highest homogeneous part of $P_{C_n}$. 
For the Casimir polynomial $C^k(x)$ used to determine $\hbar(\mu_j)$, the
assumption
$
P_{C_k}^{[d_k]}(\mu_\infty)\neq0
$
and the condition
$
\bigl(\hbar(\mu_j)\bigr)^{d_k}C_k(\mu_j)=\lambda_k
$
imply that
\begin{align}
\lim_{j \rightarrow \infty} \hbar(\mu_j)\| \mu_j \| = \alpha,
\end{align}
where $\alpha$ is determined by
\begin{align}
\alpha^{d_k}P_{C_k}^{[d_k]}(\mu_\infty)=\lambda_k.
\end{align}
Consequently, for every $n=1,\ldots,r$,
\begin{align}
\bigl(\hbar(\mu_j)\bigr)^{d_n}C_n(\mu_j)
&=
\bigl(\hbar(\mu_j)\| \mu_j \|\bigr)^{d_n}
\| \mu_j \|^{-d_n}C_n(\mu_j)
\rightarrow
\alpha^{d_n}P_{C_n}^{[d_n]}(\mu_\infty) = \lambda_n.
\end{align}
(In particular, if
$
P_{C_n}^{[d_n]}(\mu_\infty)=0,
$
then the corresponding limiting Casimir value is zero.)

\end{proof}

\begin{rem}
The map
$
q_{A/I,\mu_n}:A_\mathfrak{g}/I(C)\to \operatorname{End}(V^{\mu_n})
$
is defined on the quotient by the single equation
$
I(C):=(C^k(x)-\lambda_k).
$
Therefore, at the level of its defining algebraic variety, it is not a quantization of the common level set of all Casimir polynomials.

Indeed, for another Casimir polynomial $C^i(x)$, the remainder obtained by division by the Gr\"{o}bner basis of the ideal $I(C)$ is not, in general, a Casimir polynomial. Thus this construction does not directly define a quantization of the common level set
\begin{align}
C^1(x)=\lambda_1,\quad \ldots,\quad C^r(x)=\lambda_r,
\end{align}
which corresponds to a coadjoint orbit.

Nevertheless, Proposition \ref{Syusoku1} implies that
it corresponds to a coadjoint orbit in the classical limit.
The representation $V^{\mu_j}$ is irreducible, and the corresponding map is not faithful. Its kernel imposes additional relations depending on the representation. 
In the large-$N$ limit, or equivalently in the limit $\hbar(\mu_j)\to 0$, these additional relations fix the limiting values of the remaining Casimir polynomials. Consequently, the effective classical limit is the coadjoint orbit determined by the limiting values
\begin{align*}
\lambda_i
=
\lim_{n\to\infty}
\bigl(\hbar(\mu_n)\bigr)^{d_i}C_i(\mu_n),
\qquad
i=1,\ldots,r.
\end{align*}
We shall examine this point further below.
\end{rem}

\subsection{$su(3)$ case}
We examine the content of the preceding subsection in greater detail, focusing particularly on the case of $su(3)$. We use the notation summarized in Appendix \ref{Ap_su(3)}.

Since $\operatorname{rank}su(3)=2$, the highest weight vector can be written as
\begin{align}
\mu = p\omega_1+q\omega_2
\end{align}
in terms of the fundamental weight vectors $\omega_1$ and $\omega_2$. In what follows, we use the Dynkin label $(p,q)$ instead of $\mu$.

By the Weyl dimension formula, the dimension of the representation space $V^{(p,q)}$ of the $(p,q)$ representation is given by
\begin{align}
\dim V^{(p,q)}
=
\frac{1}{2}(p+1)(q+1)(p+q+2).
\label{vpq}
\end{align}
The eigenvalue of the quadratic Casimir operator is
\begin{align}
C_2(p,q)
=
\frac{1}{3}(p^2+q^2+pq)+(p+q),
\label{su3_2nd_Cas}
\end{align}
and the eigenvalue of the cubic Casimir operator is given by
\begin{align}
C_3(p,q)
=
\frac{1}{18}(p-q)(2p+q+3)(p+2q+3).
\label{su3_3_Cas}
\end{align}
According to the notation used in the general theory, these should be denoted by $C_i(p,q)$ $(i=1,2)$. However, in the case of $su(3)$, we write them as $C_2(p,q)$ and $C_3(p,q)$, respectively, in order to distinguish clearly between the quadratic and cubic Casimir operators.
The overall numerical factors are matters of normalization and are not essential.

As a sequence of representations $(p_n,q_n)$ for which $\dim V^{(p_n,q_n)}$ diverges, it is sufficient to consider the case
\begin{align}
\lim_{n\to\infty} (p_n + q_n)=\infty .
\end{align}
We fix each $\hbar_{(p_n,q_n)}$ by imposing
\begin{align}
\bigl(\hbar_{(p_n,q_n)}\bigr)^2 C_2(p_n,q_n)=r^2
\end{align}
with a positive constant $r$. Namely,
\begin{align}
\hbar_{(p_n,q_n)}
=
\frac{r}{\sqrt{C_2(p_n,q_n)}} .
\end{align}
Since at least one of $p_n$ and $q_n$ is nonzero, we may assume $q_n\neq 0$ in the following discussion. 
(The case $p_n\neq 0$ can be treated similarly and is therefore omitted.) Writing
\begin{align}
t_n=\frac{p_n}{q_n},
\qquad
0\le t_n<\infty ,
\end{align}
we obtain
\begin{align}
\lim_{n\to\infty} R_{3/2}(t_n)
:=
\lim_{n\to\infty}
\frac{C_3(p_n,q_n)}{\bigl(C_2(p_n,q_n)\bigr)^{3/2}}
=
\frac{\sqrt{3}}{6}
\lim_{n\to\infty}
\frac{(t_n-1)(2t_n+1)(t_n+2)}
{(t_n^2+t_n+1)^{3/2}} .
\end{align}
The function $R_{3/2}(t)$ is monotonically increasing for $0\le t$, and satisfies
\begin{align}
R_{3/2}(0)=-\frac{\sqrt{3}}{3},
\qquad
\lim_{t\to\infty}R_{3/2}(t)=\frac{\sqrt{3}}{3}.
\end{align}
Hence
\begin{align}
-\frac{\sqrt{3}}{3}
\le
R_{3/2}(t_n)
\le
\frac{\sqrt{3}}{3}.
\end{align}

We now fix a sequence of representations as in Proposition \ref{Syusoku1}. Suppose that
\begin{align}
\left(
\frac{p_n}{\sqrt{p_n^2+q_n^2}},
\frac{q_n}{\sqrt{p_n^2+q_n^2}}
\right)
\rightarrow
(p_\infty,q_\infty).
\end{align}
Then
\begin{align}
R
:=
\lim_{n\to\infty}
\frac{C_3(p_n,q_n)}
{\bigl(C_2(p_n,q_n)\bigr)^{3/2}}
=
\frac{\sqrt{3}}{6}
\frac{
(p_\infty-q_\infty)(2p_\infty+q_\infty)(p_\infty+2q_\infty)
}{
(p_\infty^2+q_\infty^2+p_\infty q_\infty)^{3/2}
}.
\end{align}

We summarize the result as follows.

\begin{proposition}\label{prop2_3}
Let $(p_n,q_n)$ be a sequence of representations such that
$\displaystyle
\lim_{n\to\infty}(p_n + q_n)=\infty$ and
$\displaystyle
\lim_{n\to\infty}t_n=t,
$
$ (0\le t\le \infty ),$
where $t_n=p_n/q_n$. 
Suppose that
$$\displaystyle
\lim_{n\to\infty}
\bigl(\hbar_{(p_n,q_n)}\bigr)^2 C_2(p_n,q_n)
=
r^2 \ \ (r>0).
$$
Then
$
\bigl(\hbar_{(p_n,q_n)}\bigr)^3 C_3(p_n,q_n)
$
converges, and
\begin{align}
\lim_{n\to\infty}
\bigl(\hbar_{(p_n,q_n)}\bigr)^3 C_3(p_n,q_n)
=
r^3 R .
\end{align}
By choosing the sequence of representations $(p_n,q_n)$, the real number $R$ can take any value in the interval
\begin{align}
-\frac{\sqrt{3}}{3}
\le
R
\le
\frac{\sqrt{3}}{3}.
\end{align}
\end{proposition}

Therefore, as will be shown below, if the sequence $(p_n,q_n)$ above is used as a sequence of irreducible representations in the construction of the quantization as a matrix regularization, then the classical limit is the coadjoint orbit determined by
\begin{align}
C^2(x)=r^2,
\qquad
C^3(x)=r^3R .
\end{align}
According to the notation used in the general theory, these Casimir polynomials should be denoted by $C^1(x)$ and $C^2(x)$. 
As in the case of the Casimir eigenvalues, we denote the quadratic Casimir polynomial by $C^2(x)$ and the cubic Casimir polynomial by $C^3(x)$.
The same convention will be used for the case of $su(2)$ in the next section. The reader should note that this shifts the indexing convention by one. This notation reflects the fact that, in the present examples, the invariant polynomials can be distinguished by their degrees.

\section{Strategy illustrated by a simple $su(2)$ Model; Fuzzy ${\mathbb R}^3$ from Fuzzy $S^2$ }\label{sect4}
For irreducible representations, there are well-established correspondences similar to those appearing in geometric quantization and deformation quantization. 
In contrast, the corresponding framework for reducible representations has not been systematically discussed. 
Therefore, in order to make the strategy transparent, it is useful first to understand it through a simple example.
As the simplest example, we first consider the  Fuzzy ${\mathbb R}^3$.\\

The  fuzzy sphere is considered in {\rm \cite{matrix1,fuzzy1}}. See {\rm \cite{matrix1,fuzzy1,fuzzyb,fuzzyc}} for details. 
In \cite{Rieffel:2021ykh}, more general and mathematically precise statements are given. 
The fuzzy ${\mathbb R}^3$ is discussed in \cite{Hammou:2001cc,Vitale:2012dz}.
\\

Let us consider ${su}(2)$ as $\mathfrak{g} $. 
Let $x_a~(1\le a\le 3)$ be commutative variables. 
$(x_1, x_2, x_3)=(x,y,z)$ is identified with the coordinates of ${\mathbb R}^3$.
The Lie-Poisson structure is defined by
\begin{align}
\{x_a, x_b\}=  \epsilon^{abc}x_c. \label{eq.poi}
\end{align}
$A_{{su}(2)}$ is given by $\mathbb{C}[x]$ with this Poisson bracket.
For an arbitrary element $f\in A_{{su}(2)}$, we write
\begin{align*}
f
=
f_0+f_a x_a+\frac{1}{2}f_{ab}x_a x_b+\cdots ,
\end{align*}
where $f_{a_1\cdots a_i}\in \mathbb{C}$ is completely symmetric with respect to the indices $a_1,\ldots,a_i$.

We emphasize that, unlike in the treatment found, for example, in Madore's textbook, we do not impose the traceless condition for arbitrary pairs of indices. Imposing this condition restricts the expansion to functions on the sphere. In the present discussion, however, we are simply considering the quantization of ${\mathbb R}^3$, and therefore we do not impose this condition.
Let $V^\mu$ be a finite-dimensional vector space over ${\mathbb C}$. 
For example, consider the two-dimensional representation $V^2={\mathbb C}^2$. 
Then the map
\[
q_2:A_{{su}(2)}\to {\rm Mat}_2({\mathbb C})
\]
is defined as a linear map by
\begin{align*}
q_2(f)
&:=
f_0 {Id}_2+f_a q_2(x_a),
\qquad
q_2(x_a):=\frac{\hbar}{2}\sigma^a ,
\end{align*}
where $\sigma^a$ are the Pauli matrices and ${Id}_k$ denotes the $k\times k$ identity matrix. \\

For $\dim V^\mu=k\ge 2$, the matrix regularization
\[
q_{k,r}:A_{{su}(2)}\to {\rm Mat}_k({\mathbb C})
\]
is defined by
\begin{align*}
q_{k,r}(f)
&:=
f_0 {Id}_k
+f_{a_1} q_{k,r}(x_{a_1})
+\frac{1}{2!}f_{a_1a_2} q_{k,r}(x_{a_1}x_{a_2})
+\cdots,\\
q_{k,r}(x_{a_1}\cdots x_{a_m})
&:=
\frac{\hbar(k,r)^m}{m!}
\sum_{\sigma\in {\rm Sym}(m)}
J_{a_{\sigma(1)}}\cdots J_{a_{\sigma(m)}} .
\end{align*}
Here $J_a$ are the generators of the $k$-dimensional irreducible representation of ${su}(2)$, namely the spin-$s$ representation with $k=2s+1$. 
The meaning of the subscript $r$ will become clear later. Each $q_{k,r}$ maps a polynomial to a $k\times k$ matrix.
The generators $J_a$ satisfy
\begin{align}
[J_a,J_b]
&=
i\epsilon^{abc}J_c,
\\
[q_{k,r}(x_a),q_{k,r}(x_b)]
&=
i\hbar(k,r)^2\epsilon^{abc}J_c
=
i\hbar(k,r)\epsilon^{abc}q_{k,r}(x_c).
\label{eq.com}
\end{align}
Here, since the basis is chosen to consist of Hermitian matrices, as is customary in physics, the structure constants differ by a factor of $i=\sqrt{-1}$.
The corresponding Casimir operator is
\begin{align}
\delta^{ab}J_aJ_b
=
\frac{1}{4}(k^2-1) {Id}_k .
\end{align}
Solving $\{ x_a , f(x) \} = 0~  (a=1,2,3)$ for a 2nd-degree homogeneous polynomial $f \in A_{su(2)}$, we obtain 
a solution as a quadratic Casimir polynomial
\begin{align}\label{su(2)casimir}
C^2(x) = \delta^{ab}x_a x_b .
\end{align}
$q_{k,r} (C^2(x)) =  \hbar^2(k,r) \delta^{ab}J_a J_{b}$ is the Casimir invariant.\\

We now show that, because of the existence of invariant polynomials, the quantization map has a nontrivial kernel and therefore cannot be faithful. For example, as in \cite{Gohara_Sako}, for a fixed radius $r$, we choose $\hbar(k,r)$ so that
\begin{align}
q_{k,r}(C^2(x))
=
\hbar^2(k,r)\delta^{ab}J_aJ_b
=
\hbar^2(k,r)\frac{1}{4}(k^2-1){Id}_k
= r^2{Id}_k .
\label{casimir_eigen_su(2)}
\end{align}
With this choice of $\hbar(k,r)$, we have
\begin{align}
q_{k,r}(C^2(x)-r^2)=0.
\end{align}
Hence $C^2(x) - r^2$ belongs to the kernel of $q_{k,r}$. Thus, if we consider only a single irreducible representation, the quantization of ${\mathbb R}^3$ cannot be faithful.
(We note, however, that a matrix regularization cannot be faithful, since it is a map into a finite-dimensional matrix algebra. Nevertheless, by introducing a filtration, for example by polynomial degree, one can obtain injectivity on each filtered component up to a given degree. What is shown here is a different phenomenon: the kernel of a matrix regularization contains elements arising from the existence of Casimir invariants.
)

The more general statement is as follows. \\
\bigskip

Let $\mathfrak{g}$ be a semisimple Lie algebra of a compact Lie group, and consider the algebra $A_\mathfrak{g}$. 
For a fixed $k$, we choose $\hbar(\mu)$ so that
\begin{align}
q_\mu(C^k(x)-\lambda_k)=0 .
\end{align}
Then it is clear that, independently of the rest of the structure of $A_\mathfrak{g}$, there is a kernel corresponding to functions on the level set
\begin{align}
C^k(x)=\lambda_k .
\end{align}
Moreover, there are further elements in the kernel, as follows.

As stated in \cite{Gohara_Sako}, we have the following proposition.
\begin{proposition}
Let $f^C(x)$ be a Casimir polynomial. Then $q_\mu(f^C(x))$ is a Casimir operator.
\end{proposition}
Since $V^\mu$ is irreducible, this Casimir operator acts as a scalar multiple of the identity by Schur's lemma. Let $\lambda_{f^C}(\hbar(\mu))$ denote the eigenvalue of $q_\mu(f^C(x))$ on $V^\mu$. Then
\begin{align}
q_\mu\bigl(f^C(x)-\lambda_{f^C}(\hbar(\mu))\bigr)=0 .
\end{align}
If $\lambda_{f^C}$ is allowed to depend on $\hbar(\mu)$, then, regardless of whether $\hbar(\mu)$ is chosen so that
$
q_\mu(C^k(x)-\lambda_k)=0
$
holds, there exist $r=\operatorname{rank}\mathfrak{g}$ independent Casimir operators. Thus, for a fixed irreducible representation, the $r$ independent Casimir polynomials give rise to kernel elements of the form
\begin{align}
C^i(x)-\lambda_i(\hbar(\mu)),
\qquad
i=1,\ldots,r .
\end{align}
In this sense, choosing a representation generally produces kernel elements associated with the Casimir polynomials.
\\

It should be noted that, although
$
C^i(x)-\lambda_i(\hbar(\mu))
$
belongs to the kernel of $q_\mu$, for a general polynomial
$g(x)\in A_{\mathfrak g}$, one has
\begin{align}
q_\mu\left(g(x)(C^i(x)-\lambda_i(\hbar(\mu)))\right)
=
0\mod\hbar(\mu).
\end{align}
\bigskip

Taking this into account, we return to the example of $su(2)$ and continue the discussion. 
The meaning of the subscript $r$ is now clear: $q_{k,r}$ denotes the quantization map associated with the irreducible representation and with the choice of $\hbar(k,r)$ satisfying (\ref{casimir_eigen_su(2)}), for $k=2,3,\ldots$.
This means that, for polynomial functions that vanish on the sphere of fixed radius $r$, their images under the quantization map tend to zero in the classical limit. Therefore, the quantization defined by such a sequence of irreducible representations should be regarded not as a quantization of ${\mathbb R}^3$ itself, but rather as a quantization of a sphere of fixed radius, namely as a fuzzy sphere. Based on this observation, we discussed quantization using reducible representations in the previous paper\cite{Gohara_Sako}.

In what follows, we explicitly carry out the construction of fuzzy ${\mathbb R}^3$ by using quantization associated with reducible representations.\\
\bigskip

As already reviewed in Subsection \ref{rev_sect6.1}, weak matrix regularization can also be constructed for reducible representations. We shall use this construction here. For $k\in{\mathbb N}$, consider the sequence
\begin{align}
r_1=\sqrt{\frac{1}{k}},
\quad
r_2=\sqrt{\frac{2}{k}},
\quad
\ldots,
\quad
r_{k^2}=\sqrt{\frac{k^2}{k}}=\sqrt{k}.
\end{align}
We define a sequence of quantization maps associated with reducible representations by
\begin{align}
q_k^R:=\bigoplus_{i=1}^{k^2} q_{k,r_i},
\qquad
k=2,3,\ldots .
\end{align}

Note that, for each $q_{k,r_i}$, in the limit $k\to\infty$ we have
$
\hbar(k,r_i)\to 0
$ and,
$
r_{k^2}\to\infty . $
Moreover, in the limit $k\to\infty$, we have
\begin{align}
|r_{i+1}-r_i|
=
\left|
\sqrt{\frac{i+1}{k}}-\sqrt{\frac{i}{k}}
\right|
\le
\left|
\sqrt{\frac{2}{k}}-\sqrt{\frac{1}{k}}
\right|
\to 0 .
\end{align}
This means that, by taking $k$ sufficiently large, the spacing between the radii of the spheres can be made arbitrarily small.

Thus the sequence of quantization maps is constructed so that, for any point in ${\mathbb R}^3$, there is a sphere in the sequence passing through an arbitrarily small neighborhood of that point. In this sense, the construction corresponds to a quantization of the decomposition of ${\mathbb R}^3$ into spheres, which becomes dense in the classical limit.


\section{The $su(3)$ case: fuzzy $S^7$}\label{sect5}

In the previous section, we considered the quantization of ${\mathbb R}^3$. 
Since irreducible representations correspond to quantizations of coadjoint orbits, we realized a matrix regularization of ${\mathbb R}^3$ by filling ${\mathbb R}^3$ densely with spheres, which are coadjoint orbits. 
In this section, we construct a more complicated example, namely fuzzy $S^7$, by using the Lie algebra $su(3)$.\\

Using the notation for $su(3)$ given in Appendix \ref{Ap_su(3)}, we first prepare a matrix regularization of ${\mathbb R}^8={su}(3)^*$. 
The coordinates $x=(x_1,\ldots,x_8)$ are equipped with the Lie-Poisson structure
\begin{align}
\{x_a,x_b\}=f_{ab}^{\ \ c}x_c .
\end{align}
Let $V^{(p,q)}$ denote the representation space of the highest-weight representation $\rho^{(p,q)}$ with Dynkin label $(p,q)$. 
The elements $\rho^{(p,q)}(T_i)$ in $\operatorname{End}(V^{(p,q)})$ have structure constants $f_{ab}^{\ \ c}$:
\begin{align*}
[\rho^{(p,q)}(T_a),\rho^{(p,q)}(T_b)]
=
i f_{ab}^{\ \ c}\rho^{(p,q)}(T_c).
\end{align*}

In this section, in order to construct fuzzy $S^7$ of radius $r$, we define $\hbar_{(p,q)}$ by using the eigenvalue $C_2(p,q)$ of the quadratic Casimir operator. Namely, we impose
$
\hbar_{(p,q)}^2 C_2(p,q)=r^2,
$
or equivalently
\begin{align}
\hbar_{(p,q)}
:=
\frac{r}{\sqrt{C_2(p,q)}} .
\end{align}
Using this $\hbar_{(p,q)}$, we rescale the generators as
\begin{align}
e_i^{(p,q)}
:=
\hbar_{(p,q)}\rho^{(p,q)}(T_i),
\qquad
i=1,2,\ldots,8 .
\end{align}
They satisfy
\begin{align}
[e_a^{(p,q)},e_b^{(p,q)}]
= i
\hbar_{(p,q)}f_{ab}^{\ \ c}e_c^{(p,q)} .
\end{align}

In the present case, the basic quantization map 
$q_{(p,q)}:A_{{su}(3)}\rightarrow \operatorname{End}(V^{(p,q)})$
in Definition \ref{def_FuzzyQ} is given by
\begin{align}
q_{(p,q)}(x_i)
=
e_i^{(p,q)}
=
\hbar_{(p,q)}\rho^{(p,q)}(T_i).
\end{align}
Then we have
\begin{align}
q_{(p,q)}(C^2(x))
=
\hbar_{(p,q)}^2 C_2(p,q)~ {Id}_{V^{(p,q)}}
=
r^2~ {Id}_{V^{(p,q)}},
\end{align}
and
\begin{align}
q_{(p,q)}(C^3(x))
=
\hbar_{(p,q)}^3 C_3(p,q) {Id}_{V^{(p,q)}} .
\end{align}
Since $\hbar$ is fixed by the condition $\hbar_{(p,q)}^2 C_2(p,q)=r^2$, the quadratic Casimir operator is mapped to a fixed constant. On the other hand, the cubic Casimir operator is mapped to a scalar depending on the representation, or equivalently on $\hbar_{(p,q)}$ through the chosen representation.

\subsection{Fuzzy $S^7$}

In this subsection, we construct fuzzy $S^7$. 
We first show that sequences of irreducible representations used in the matrix regularization of $S^7$ give rise, in the classical limit, to coadjoint orbits. 
We then show that, by combining these sequences through reducible representations, the resulting coadjoint orbits fill $S^7$ densely.\\

To consider a quantization that yields fuzzy $S^7$,
we consider the ideal of $A_{su(3)}$ given by
\begin{align}
I(C):=(C^2(x)-r^2).
\end{align}
In this case, the Gr\"{o}bner basis $G$ is also given by
$
G=\{C^2(x)-r^2\}.
$

Next, we define the matrix regularization
\begin{align}
q_{A_{{su}(3)}/I,(p,q)}
:
A_{su(3)}/I(C)
\rightarrow
\operatorname{End}(V^{(p,q)})
\end{align}
associated with one irreducible representation by (\ref{q_AImu_general}). 
Namely, for any $f\in A_{{su}(3)}$, let $r_f$ denote the remainder obtained by dividing $f$ by $C^2(x)-r^2$. 
We assume that $\deg r_f\le n_{(p,q)}$. If necessary, this condition can always be satisfied by taking $\dim V^{(p,q)}$ sufficiently large compared with the degree of $r_f$. 
Then
\begin{align}
q_{A_{su(3)}/I,(p,q)}([f])
=
q_{(p,q)}(r_f).
\end{align}

Since there exists a polynomial $d(x)$ such that
\begin{align}
C^3(x)=d(x)(C^2(x)-r^2)+r_{C^3(x)},
\end{align}
using (\ref{asym_hom_mu}) and $q_{(p,q)}(C^2(x)-r^2) =0$  we have
\begin{align}
q_{A_{{su}(3)}/I,(p,q)}([C^3(x)])
&=
q_{(p,q)}(r_{C^3(x)}) \notag\\
&=
q_{(p,q)}(C^3(x))
+
\hbar_{(p,q)}D(e^{(p,q)}).
\end{align}
Here $D(e^{(p,q)})$ is an expression of degree at most two in $e^{(p,q)}$, namely
\begin{align}
D(e^{(p,q)})
=
\sum_{i,j}a_{ij}e_i^{(p,q)}e_j^{(p,q)},
\qquad
a_{ij}\in{\mathbb C},
\end{align}
where the coefficients $a_{ij}$ do not depend on $\hbar_{(p,q)}$. 
In other words
$
\hbar_{(p,q)}D(e^{(p,q)})
$
is 
$\tilde{O}( \hbar_{(p,q)} )$.
Using (\ref{su3_3rd_Casimir}), we obtain
\begin{align}
q_{(p,q)}(C^3(x))
=
\hbar_{(p,q)}^3 C_3(p,q)\operatorname{Id}_{V^{(p,q)}}.
\end{align}
We have already seen in Proposition \ref{prop2_3} that this converges to $r^3R$ in the limit.
Moreover, depending on the choice of the limit of
$
\frac{p_n}{q_n}=t_n,
$
the quantity $R$ can take any value in the interval
$
-\frac{\sqrt{3}}{3}
\le
R
\le
\frac{\sqrt{3}}{3}.
$
This means that, in the limit \begin{align} \dim V^{(p,q)}\to\infty, \end{align} or equivalently $\hbar_{(p,q)}\to 0$, the algebra asymptotically approaches the algebra of polynomial functions on the algebraic variety determined by \begin{align} \sum_{i=1}^8 x_i^2=r^2, \qquad C^3(x)=r^3R, \label{lim_ire_rep} \end{align} namely, on a single coadjoint orbit.

To explain this correspondence in more detail, we use the map
\begin{align}
\phi_\mu:\operatorname{End}(V^\mu)\rightarrow A_{\mathfrak g}
\end{align}
introduced in \cite{Gohara_Sako}. 
For each PBW(Poincar\'{e}-Birkhoff-Witt)-type symmetrized element retained in the chosen basis, this map is defined by
\begin{align}
\phi_\mu\left(e^{(\mu)}_{(i_1,\ldots,i_k)}\right)
=
x_{i_1}\cdots x_{i_k}. \label{asy_hom}
\end{align}

In the present case, we use
\begin{align}
\phi_{(p,q)}:
\operatorname{End}(V^{(p,q)})
\rightarrow
A_{{su}(3)}
\end{align}
to establish this correspondence. This map was referred to as an asymptotic algebra homomorphism.
When $\dim V^{(p,q)}$ is sufficiently large, 
for any $g(x) \in A_{{su}(3)}$ 
the asymptotic algebra homomorphism satisfies
$
\phi_{(p,q)}\circ q_{(p,q)}( g(x)(C^i(x)-\lambda_i (\hbar_{(p,q)})))
= 0
\mod \hbar_{(p,q)}.
$
Therefore, as $\hbar_{(p,q)}\to 0$, the asymptotic algebra homomorphism relates the algebra generated by $e^{(p,q)}$ to the coordinate ring of the corresponding coadjoint orbit,
$
\mathbb{C}[x_1,\ldots,x_8]/(C^2(x)-r^2,C^3(x)-r^3R).
$
The multiplicative defect of this map vanishes in the classical limit. 
However, the map $\phi_{(p,q)}$ need not be surjective onto the coordinate ring of the coadjoint orbit, even when restricted to elements of low polynomial degree. 
Indeed, finite-dimensional representations may satisfy additional representation-dependent relations that do not follow from the Casimir constraints. 
As a result, some PBW-type symmetrized elements that would correspond to 
monomials in the coordinate ring fail to be linearly independent and therefore cannot be included in a basis of $\operatorname{End}(V^{(p,q)})$. 
The corresponding monomials may consequently be absent 
from the image of $\phi_{(p,q)}$.
Thus, although the multiplicative defect of $\phi_{(p,q)}$ vanishes in the classical limit, the map does not necessarily become surjective degree by degree as $\hbar_{(p,q)}\to 0$. Accordingly, a degreewise asymptotic isomorphism with the full coordinate ring is not asserted in general.
If, for a particular sequence of representations, no additional representation-dependent relations persist in any fixed PBW degree as $\hbar_{(p,q)}\to 0$, then the correspondence becomes asymptotically one-to-one degree by degree.
Thus, by choosing a suitable sequence of representations, the correspondence may approach an isomorphism. We do not, however, pursue this issue in this paper. \\
\bigskip

We now return to the general case and introduce several definitions.
Consider
$
\phi_\mu\circ q_\mu
:
\mathbb{C}[x]
\longrightarrow
\mathbb{C}[x]
$
such that there exist polynomials
$f^1_\hbar(x),\ldots,f^n_\hbar(x)$ satisfying
\begin{align}
q_\mu\bigl(f^a_\hbar(x)\bigr)=0,
\qquad
a=1,\ldots,n.
\end{align}
For every
$g(x)\in\mathbb{C}[x]$ we have
\begin{align}
\phi_\mu\circ q_\mu \bigl(g(x)f^a_\hbar(x)\bigr)
&=
\phi_\mu ( q_\mu(g(x))  q_\mu\bigl(f^a_\hbar(x)\bigr) )
\mod\hbar
\notag\\
&=
0\mod\hbar,
\end{align}
for sufficiently large $\dim V^{\mu}$.
Consequently, for the ideal
\begin{align}
I_\hbar
:=
\bigl(f^1_\hbar(x),\ldots,f^n_\hbar(x)\bigr),
\end{align}
we have
$
\phi_\mu\circ q_\mu(I_\hbar)=0\mod\hbar.
$
Thus, modulo terms proportional to $\hbar$, the map $\phi_\mu\circ q_\mu$ factors
through the quotient algebra $\mathbb{C}[x]/I_\hbar$.

Suppose moreover that, as $\hbar(\mu)\to0$, the polynomials
$f^a_\hbar(x)$ converge to polynomials $f^a_0(x)$ independent of
$\hbar$. In what follows, we shall say that the sequence converges in
the sense of relations to the algebraic variety
\begin{align}
M
:=
\left\{
x
\ \middle|\
f^1_0(x)=0,\ldots,f^n_0(x)=0
\right\},
\end{align}
or, equivalently, that it has $M$ as its classical limit in the sense
of relations.\\

Let us return to the topic of $su(3)$.
For a sequence of irreducible representations such that
$
\frac{p_n}{q_n}=t_n\rightarrow t,
$
where the case $t=\infty$ is also included, the classical limit converges to the coadjoint orbit determined by (\ref{lim_ire_rep}) in the sense of relations. In the present case, such coadjoint orbits are of the types
$SU(3)/T^2$
and 
${\mathbb C}P^2$.\\

\begin{rem}
In the definition of convergence in the sense of relations, it should
be noted that $\phi_\mu$ is not necessarily essential. Rather, convergence
to a coadjoint orbit in the classical limit occurs because the
components of $\ker q_\mu$ associated with the Casimir polynomials
converge to the defining relations of a certain coadjoint orbit. As a
result, the image of $\phi_\mu$ approaches the algebra of functions on that
orbit. Thus, although $\phi_\mu$ describes the corresponding classical
limit, it is not the cause of this convergence. In the examples
considered here, the essential ingredient is the choice of the sequence
of representations used to construct the sequence of maps $q_\mu$.
\end{rem}

We have shown that sequences of irreducible representations can be chosen so that their classical limits are coadjoint orbits. 
It remains to show that, by taking a sequence of suitable collections of irreducible representations and combining them into reducible representations, 
the corresponding coadjoint orbits fill $S^7$ densely.

We show that, for every point of $S^7$, there exists a solution of (\ref{lim_ire_rep}) containing that point, or equivalently, that there exists a sequence of irreducible representations whose corresponding coadjoint orbits approach that point.

The equations in (\ref{lim_ire_rep}) are expressed in terms of the Casimir polynomials ${\mathrm{tr}}\,X^2$ and ${\mathrm{tr}}\,X^3$ and constants, and are therefore invariant under the action of $SU(3)$. Hence, by conjugation under $SU(3)$, it is sufficient to restrict the discussion to the Cartan subalgebra. Namely, we may set
\begin{align}
x_1=x_2=x_4=x_5=x_6=x_7=0
\end{align}
and consider
\begin{align}
x_3^2+x_8^2&=r^2,\\
-\frac{\sqrt{3}}{3}x_8^3+\sqrt{3}x_3^2x_8&=r^3R.
\end{align}
The points on the resulting circle in the Cartan subalgebra can be parametrized as
\begin{align}
x_8=r\cos\theta,
\qquad
x_3=r\sin\theta,
\qquad
\theta\in{\mathbb R}.
\end{align}
Then the left-hand side of the second equation becomes
\begin{align*}
-\frac{\sqrt{3}}{3}x_8^3+\sqrt{3}x_3^2x_8
&=
-\frac{\sqrt{3}}{3}r^3\cos\theta
\left(\cos^2\theta-3\sin^2\theta\right)\\
&=
-\frac{\sqrt{3}}{3}r^3\cos 3\theta.
\end{align*}
Therefore, as $\theta$ varies, the value of the cubic Casimir ranges over the entire interval
\begin{align}
-\frac{\sqrt{3}}{3}r^3
\le
r^3R
\le
\frac{\sqrt{3}}{3}r^3.
\end{align}
As shown above, for every value in this interval, there exists a sequence of irreducible representations for which the corresponding normalized cubic Casimir converges to that value.

Consequently, for every point of $S^7$, there exists a sequence of coadjoint orbits associated with irreducible representations that approaches that point. Equivalently, the coadjoint orbits arising from such sequences fill $S^7$ densely. \\
\begin{definition}
Let $M$ be an algebraic variety that admits a decomposition into a disjoint union
$
M=\coprod_i M_i.
$
Suppose that the matrix regularizations
$q_{A/I,\mu}^a$ introduced in Definition \ref{def_2_10}, each
associated with an irreducible component of the representation,
converge to the spaces $M_{i_\mu^a}$ in the sense of relations.
Suppose further that the resulting sequence of spaces
$M_{i_\mu^a}$ fills $M$ densely.
Here, ``dense'' is understood with respect to the Euclidean topology.
Then we say that 
$q_{A/I,\mu}^R:
A_\mathfrak{g}/I(C)
\rightarrow
\operatorname{End}(V^\mu)$ 
in Definition \ref{def_2_10}
has $M$ as its classical limit in the sense of weak matrix regularization.
\end{definition}

From the above discussion, the following is obtained.

\begin{theorem}
There exists a weak matrix regularization of $S^7$ constructed from a sequence of reducible representations whose classical limit, in the sense of weak matrix regularization, is $S^7$.
\end{theorem}

A matrix regularization constructed from reducible representations whose classical limit is $S^7$ can, for example, be given explicitly as follows:
\begin{align}
q_{A_{{su}(3)}/I,k}^{R}
:=
\bigoplus_{\substack{p,q\in{\mathbb Z}_{\geq 0}\\ p+q=k}}
q_{A_{{su}(3)}/I,(p,q)},
\qquad
k=1,2,3,\ldots .
\end{align}
It follows from (\ref{su3_dim}) that, for every pair $(p,q)$ satisfying $p+q=k$,
\begin{align}
\dim V^{(p,q)}\rightarrow\infty
\end{align}
as $k\to\infty$. Moreover, $\hbar_{(p,q)}\to 0$ in this limit.

Therefore, as discussed above, it remains to show that the ratios $p/q$, or equivalently $q/p$, of the nonnegative integer pairs satisfying $p+q=k$ become dense in the extended interval $[0,\infty]$ as $k\to\infty$. Equivalently, by applying the continuous map $\arctan$, it is sufficient to show that the corresponding angles become dense in $[0,\pi/2]$.

The adjacent lattice points satisfying $p+q=k$ can be written as
\begin{align}
(k-i,i)
\qquad\text{and}\qquad
(k-i-1,i+1),
\qquad
i=0,1,\ldots,k-1.
\end{align}
The corresponding angles are
\begin{align}
\arctan\left(\frac{k-i}{i}\right)
\qquad\text{and}\qquad
\arctan\left(\frac{k-i-1}{i+1}\right),
\end{align}
where, for $i=0$, we define
$\displaystyle
\arctan\left(\frac{k-i}{i}\right)=\frac{\pi}{2}.
$
Thus, it is sufficient to show that the maximum gap between consecutive terms in the sequence of angles converges to zero.  Indeed, a direct calculation with the sum identity
for $\tan$ gives
\begin{align*}
0
&\leq
\max_{0\leq i\leq k-1}
\left|
\arctan\left(\frac{k-i}{i}\right)
-
\arctan\left(\frac{k-i-1}{i+1}\right)
\right| 
\leq
\frac{2k}{k^2-1},
\end{align*}
for $k\geq 2$. Since
$ \displaystyle
\frac{2k}{k^2-1}\rightarrow 0
$
as $k\to\infty$, these angles become dense in $[0,\pi/2]$.

Consequently, the corresponding ratios $p/q$ become dense in $[0,\infty]$. Since the normalized cubic Casimir is a continuous function of this ratio and takes all values in the interval
$ \displaystyle
-\frac{\sqrt{3}}{3}
\leq
R
\leq
\frac{\sqrt{3}}{3},
$
the coadjoint orbits associated with the irreducible components of $q_{A_{{su}(3)}/I,k}^{R}$ fill $S^7$ densely in the classical limit. In this way, we obtain a matrix regularization whose classical limit is $S^7$ in the sense of weak matrix regularization.

\section{Generalization}\label{sect6}
We now generalize the preceding discussion. \\



Let $\mathfrak{g}$ be a $d$-dimensional compact semisimple Lie algebra. Let
$
C^n(x),\
n=1,2,\ldots,r=\operatorname{rank}\mathfrak{g},
$
be Casimir polynomials on ${\mathbb R}^d$, where $C^n(x)$ is homogeneous of degree $d_n$ and they are algebraically independent generators
of invariant polynomials.
As is known, a coadjoint orbit is determined by fixing the values of all the Casimir polynomials. For
\begin{align}
\bigl(C^1(x),\ldots,C^r(x)\bigr)
=
(c_1,\ldots,c_r)
=:
{\mathbf c},
\end{align}
we denote by $O_{\mathbf c}$ the coadjoint orbit corresponding to ${\mathbf c}$. Using an invariant inner product to identify $\mathfrak{g}^*$ with $\mathfrak{g}$, we write
\begin{align}
O_{\mathbf c}
:=
\left\{
X=x_i e_i\in\mathfrak{g}
\ \middle|\
C^1(x)=c_1,\ldots,C^r(x)=c_r
\right\}.
\end{align}
About this description, see, for example, the discussion in \cite{Lledo}.

As a space containing $O_{\mathbf c}$, we fix only
\begin{align}
(a_1,\ldots,a_p)
=
(c_{i_1},\ldots,c_{i_p})
=:
{\mathbf a},
\qquad
p\leq r,
\end{align}
and consider
\begin{align}
S_{\mathbf a}
:=
\left\{
X=x_i e_i\in\mathfrak{g}
\ \middle|\
\bigl(C^{i_1}(x),\ldots,C^{i_p}(x)\bigr)
=
\bigl(c_{i_1},\ldots,c_{i_p}\bigr)
=
{\mathbf a}
\right\}.
\end{align}
In what follows, we denote this index set by
$
I_{\mathbf a}:=\{i_1,\ldots,i_p\},
$
and denote the complementary index set by
$
I_{\mathbf b}
:=
\{1,2,\ldots,r\}\setminus I_{\mathbf a}
=
\{j_1,\ldots,j_{r-p}\}.
$

For fixed ${\mathbf a}$, let $K_{\mathbf a}$ denote the set of all possible values of the remaining $r-p$ invariant polynomials:
\begin{align}
K_{\mathbf a}
:=
\left\{
{\mathbf b}
=
(b_1,\ldots,b_{r-p})\in{\mathbb R}^{r-p}
\ \middle|\
\begin{array}{l}
\text{there exists }x\in\mathfrak{g}\text{ such that}\\
C^{i_\ell}(x)=a_\ell
\quad
(\ell=1,\ldots,p),\\
C^{j_m}(x)=b_m
\quad
(m=1,\ldots,r-p)
\end{array}
\right\}.
\end{align}
Let $O_{{\mathbf a},{\mathbf b}}$ denote the coadjoint orbit determined by the complete set of invariant values $({\mathbf a},{\mathbf b})$. Then
$\displaystyle 
S_{\mathbf a}
=
\coprod_{{\mathbf b}\in K_{\mathbf a}}
O_{{\mathbf a},{\mathbf b}}.
$

Indeed, every point of $S_{\mathbf a}$ is an element of the Lie algebra and therefore belongs to a coadjoint orbit, where we identify $\mathfrak{g}$ with $\mathfrak{g}^*$ by means of an invariant inner product. Consequently, by allowing ${\mathbf b}$ to range over all values in $K_{\mathbf a}$, the space $S_{\mathbf a}$ is reconstructed as the disjoint union of the corresponding coadjoint orbits.

\begin{lemma}\label{lem5_0}
Let $\mathfrak{g}$ be a $d$-dimensional compact semisimple Lie algebra.
Let
$
C^n(x),
\
n=1,2,\ldots,r=\operatorname{rank}\mathfrak{g},
$
be algebraically independent generators of the algebra of invariant polynomials on $\mathfrak{g}$.
Choose an index set
$
I_{\mathbf a}:=\{i_1,\ldots,i_p\},
\
p\leq r,
$
and let
$
I_{\mathbf b}
:=
\{1,\ldots,r\}\setminus I_{\mathbf a}
=
\{j_1,\ldots,j_{r-p}\}.
$
For fixed values
$
{\mathbf a}:=(a_1,\ldots,a_p),
$
define the subvariety 
$ S_{\mathbf a} $ and $K_{\mathbf a}$ as above.
For each ${\mathbf b}\in K_{\mathbf a}$, let
$O_{{\mathbf a},{\mathbf b}}$ denote the coadjoint orbit determined by the complete set of Casimir values $({\mathbf a},{\mathbf b})$. Then
\begin{align}
S_{\mathbf a}
=
\coprod_{{\mathbf b}\in K_{\mathbf a}}
O_{{\mathbf a},{\mathbf b}}.
\end{align}
\end{lemma}

As a simple special case in which ${\mathbf a}$ has only one component,
namely ${\mathbf a}=(\lambda_k)$, we obtain the following result.

\begin{lemma}\label{lem5_1}
Assume the same conditions and notations as in Lemma~\ref{lem5_0}.
Consider the algebraic variety defined by
$
S_k
:=
\left\{
x\in{\mathbb R}^d
\ \middle|\
C^k(x)-\lambda_k=0
\right\}.
$
Let $K_k$ denote the set of all realizable values of the complete set of
Casimir polynomials subject to the condition that the $k$-th value is
$\lambda_k$:
\begin{align}\label{k_k}
K_k
:=
\left\{
{\mathbf c}=(c_1,\ldots,c_r)\in{\mathbb R}^r
\ \middle|\
\begin{array}{l}
\text{there exists }x\in{\mathbb R}^d\text{ such that}\\
C^n(x)=c_n
\quad
(n=1,\ldots,r),\
c_k=\lambda_k
\end{array}
\right\}.
\end{align}
For each ${\mathbf c}\in K_k$, let $O_{\mathbf c}$ denote the coadjoint
orbit determined by the complete set of Casimir values ${\mathbf c}$.
Then
\begin{align}
S_k
=
\coprod_{{\mathbf c}\in K_k}O_{\mathbf c}.
\end{align}
\end{lemma}

Using Lemma \ref{lem5_1}, we obtain the following theorem.

\begin{theorem}
Let $\mathfrak{g}$ be a $d$-dimensional compact semisimple Lie algebra. Let 
Casimir polynomials 
\begin{align}
C^n(x),
\qquad
n=1,2,\ldots,r=\operatorname{rank}\mathfrak{g},
\end{align}
with $\deg C^n=d_n$
be algebraically independent generators of the algebra of invariant polynomials on $\mathfrak{g}$.

For a fixed $k$ and non-zero $\lambda_k \in \mathbb{R}$, consider the algebraic variety
\begin{align}
S_k
:=
\left\{
x\in{\mathbb R}^d
\ \middle|\
C^k(x)-\lambda_k=0
\right\}
\end{align}
and the Lie-Poisson algebra $A_{\mathfrak{g}}/I(C)$ defined on $S_k$.
Here, we assume that $\lambda_k$ is an admissible value of $C^k$, so that
$S_k\neq\emptyset$.
Then there exists a sequence of reducible representations
$\{V^{\mu_j}\}_{j=1}^{\infty}$ and a weak matrix regularization
\begin{align}
q^R_{A_{\mathfrak{g}}/I,\mu_j}
:
A_{\mathfrak{g}}/I(C)
\rightarrow
\operatorname{End}(V^{\mu_j}),
\qquad
j=1,2,\ldots,
\end{align}
whose commutative limit is $S_k$  in the sense of weak matrix regularization.
\end{theorem}
\begin{proof}
i) We first consider an irreducible representation $V^\mu$. We choose
$\hbar(\mu)$ so that
\begin{align}
q_\mu(C^k(x))
=
\lambda_k\operatorname{Id}_{\mu},
\end{align}
and take $\dim V^\mu$ sufficiently large.

Let $r_{C^i}(x)$ be the remainder obtained by dividing $C^i(x)$ by the
Gr\"{o}bner basis $\{C^k(x)-\lambda_k\}$. Thus, there exists a polynomial
$g_i(x)$ such that
\begin{align}
C^i(x)
=
g_i(x)\bigl(C^k(x)-\lambda_k\bigr)
+
r_{C^i}(x).
\end{align}
Let $C_i(\mu)$ denote the eigenvalue of the Casimir operator corresponding
to $C^i(x)$ on $V^\mu$. Then
\begin{align}
q_\mu(C^i(x))
=
\bigl(\hbar(\mu)\bigr)^{d_i}
C_i(\mu)\operatorname{Id}_{\mu}.
\end{align}
Therefore,
\begin{align}
q_{A_{\mathfrak{g}}/I,\mu}
\bigl([C^i(x)]\bigr)
&=
q_\mu\bigl(r_{C^i}(x)\bigr)
\notag\\
&=
\bigl(\hbar(\mu)\bigr)^{d_i}
C_i(\mu)\operatorname{Id}_{\mu}
-
q_\mu\left(
g_i(x)\bigl(C^k(x)-\lambda_k\bigr)
\right).
\end{align}
Using (\ref{asym_hom_mu}), we have
\begin{align}
q_\mu\left(
g_i(x)\bigl(C^k(x)-\lambda_k\bigr)
\right)
&=
q_\mu(g_i(x))
q_\mu\bigl(C^k(x)-\lambda_k\bigr)
+
\hbar(\mu)P,
\end{align}
where $P = \sum_i^D c_i(\hbar(\mu)) E_i  \in T_\mu$. Since
\begin{align}
q_\mu\bigl(C^k(x)-\lambda_k\bigr)=0,
\end{align}
it follows that
\begin{align}
q_{A_{\mathfrak{g}}/I(C^k),\mu}
\bigl([C^i(x)]\bigr)
=
\bigl(\hbar(\mu)\bigr)^{d_i}
C_i(\mu)\operatorname{Id}_{\mu}
+
\hbar(\mu) P.
\end{align}
The second term vanishes in the classical limit.

ii) We now show that all possible coadjoint orbits arise in this way. 



Let
$\mathfrak{h}_+^*$ be a closed dominant Weyl chamber. By the Chevalley
restriction theorem,
\begin{align}
\left\{
\bigl(C^1(x),\ldots,C^r(x)\bigr)
\ \middle|\
x= x_i e_i \in\mathfrak{g}
\right\}
=
\left\{
\bigl(C^1(h),\ldots,C^r(h)\bigr)
\ \middle|\
h\in\mathfrak{h}_+^*
\right\}.
\end{align}
On the other hand, by the Harish--Chandra isomorphism, the Casimir
eigenvalues are expressed as
$
C_i(\mu)=P_{C_i}(\mu+\rho).
$
Fix an arbitrary $\mathbf{c}\in K_k$ in (\ref{k_k}). 
There exists $h_{\mathbf{c}}\in\mathfrak{h}_+^*$ such that 
$
C^i(h_{\mathbf{c}})=c_i, \  (i=1,\ldots,r ). 
$
In particular, $ {C}^k(h_{\mathbf{c}})=\lambda_k. $ 
Here, we apply a similar argument as in Proposition \ref{Syusoku1}.
Since the dominant integral highest weights form a lattice in $\mathfrak{h}_+^*$, 
one can choose a sequence of integral highest weights $\{\mu_j\}$ and a sequence of positive real numbers 
$t_j\to\infty$ such that 
\begin{align*} 
\frac{\mu_j+\rho}{t_j} \rightarrow h_{\mathbf{c}}. 
\end{align*} 
Remember that $ (\hbar(\mu_j))^{d_k}C_k(\mu_j)=\lambda_k $, i.e.
$\displaystyle  \hbar(\mu_j) = \left( \frac{\lambda_k}{C_k(\mu_j)} \right)^{1/d_k}$.
Moreover, since 
$\displaystyle \lim_{j \rightarrow \infty}
t_j^{-d_k}C_k(\mu_j) = {C}^k(h_{\mathbf{c}}) = \lambda_k, 
$
we have 
$\displaystyle \lim_{j \rightarrow \infty}
\hbar(\mu_j) t_j = 
\lim_{j \rightarrow \infty} \left( \frac{\lambda_k} {t_j^{-d_k}C_k(\mu_j)} \right)^{1/d_k} 
= 1. 
$
Consequently, 
\begin{align*} 
\lim_{j \rightarrow \infty}\hbar (\mu_j) (\mu_j+\rho) = 
\lim_{j \rightarrow \infty}(\hbar(\mu_j) t_j) \frac{\mu_j+\rho}{t_j} 
= h_{\mathbf{c}}. 
\end{align*} 
Therefore, for every $i=1,\ldots,r$, 
\begin{align*} 
(\hbar(\mu_j))^{d_i}C_i(\mu_j) &= (\hbar_jt_j)^{d_i} t_j^{-d_i}P_{C_i}(\mu_j+\rho) 
\rightarrow {C}^i(h_{\mathbf{c}}) = c_i. 
\end{align*} 
Thus, for every admissible tuple 
$\mathbf{c}=(c_1,\ldots,c_r)\in K_k$, there exist sequences 
$\{\mu_j\}$ and $\{\hbar (\mu_j)\}$ such that 
\begin{align} 
(\hbar(\mu_j))^{d_k}C_k(\mu_j) &= \lambda_k, \\ \lim_{j\to\infty} (\hbar(\mu_j))^{d_i}C_i(\mu_j) &= c_i, \qquad i=1,\ldots,r. 
\end{align} 
Hence, every coadjoint orbit contained in $S_k$ is obtained as the classical limit, in the sense of relations, of a suitable sequence of irreducible matrix regularizations.

iii) Since $S_k$ is a subset of a finite-dimensional Euclidean space,
choose a countable dense subset
$
\{x_\ell\}_{\ell=1}^{\infty}\subset S_k,
$
and let $\mathcal{O}_\ell$ be the coadjoint orbit through $x_\ell$.
By part ii), for each $\ell$ there exists a sequence of irreducible
quantizations converging to $\mathcal{O}_\ell$ in the sense of relations.

Choose a sequence of positive integers $\{m_j\}$ such that
$m_j\to\infty$. For each $j$ and each $a=1,\ldots,m_j$, choose a
sufficiently large member of the sequence corresponding to
$\mathcal{O}_a$, and denote it by
$
q^a_{A_{\mathfrak g}/I,\mu_j}
$
with representation space $V_{\mu_j}^a$, so that
\begin{align}
\max_{1\leq a\leq m_j}
\left|\hbar_a(\mu_j)\right|
<
\frac{1}{j},
\qquad
\min_{1\leq a\leq m_j}
\dim V_{\mu_j}^a
>
j.
\end{align}
Since $x_a\in\mathcal{O}_a$ and
$\{x_\ell\}_{\ell=1}^{\infty}$ is dense in $S_k$, the finite unions
$\displaystyle
\bigcup_{a=1}^{m_j}\mathcal{O}_a
$
fill $S_k$ densely as $j\to\infty$.
Let
\begin{align}
V^{\mu_j}
:=
\bigoplus_{a=1}^{m_j}V_{\mu_j}^a,
\qquad
q^R_{A_{\mathfrak g}/I,\mu_j}
:=
\bigoplus_{a=1}^{m_j}
q^a_{A_{\mathfrak g}/I,\mu_j}.
\end{align}
The resulting sequence is a weak matrix regularization whose
commutative limit is $S_k$ in the sense of weak matrix regularization.
\end{proof}





\section{Summary}\label{sect7}


In this paper, we investigated the classical limits of matrix regularizations of algebraic varieties defined by fixing one Casimir polynomial of a compact semisimple Lie algebra. 
We showed that, for every coadjoint orbit, there exists a
matrix regularization constructed from a sequence of irreducible
representations that converges to that orbit in the sense of relations.
Using this result, we further showed that there exist matrix regularizations constructed from sequences of reducible representations whose classical limits, in the sense of weak matrix regularization, are the algebraic varieties defined by the Casimir constraint imposed on the domain.


In the present setting, convergence to an algebraic variety in the sense of weak matrix regularization means the following. 
The reducible representations decompose into direct sums of irreducible
representations. 
Each sequence of irreducible components gives rise to
a matrix regularization that converges, in the sense of relations, 
to a coadjoint orbit contained in the algebraic variety defining the
domain.
 These coadjoint orbits are chosen so that they fill the algebraic variety densely. 
 In this sense, we constructed weak matrix regularizations of algebraic varieties defined by Casimir polynomials.

As a concrete example, we constructed fuzzy $S^7$. 
The same construction based on a quadratic Casimir polynomial yields fuzzy $S^{n^2-2}$ from a Lie algebra of type $A_{n-1}$, fuzzy $S^{n(2n+1)-1}$ 
from one of type $B_n$, and analogous fuzzy spheres from the other compact simple Lie algebras. 
More generally, by choosing a positive-definite quadratic Casimir polynomial, 
the construction can be applied to any compact semisimple Lie algebra. Algebraic varieties defined by Casimir polynomials of degree greater than two can also be treated by the general method developed in Section \ref{sect6}.\\

An important direction for future work is a more precise analysis of
convergence in the sense of relations. 
In particular, it is important
to determine under what conditions an asymptotic algebra homomorphism
$\phi_\mu$ can also become asymptotically surjective with respect to the PBW
filtration. For suitable sequences, when the matrix and polynomial
basis elements are organized according to their PBW degrees, the
maximal PBW degree up to which $\phi_\mu$ gives a bijective
correspondence may increase as $\dim V^\mu$ increases. 
Developing a
general approach to this phenomenon remains an open problem.

Another important direction concerns connections with physics. In order to understand classical solutions of matrix models, it is necessary to investigate how the present framework should be extended to noncompact Lie algebras and noncompact spaces. This problem has not yet been studied within the present framework. Moreover, although our method provides a systematic way to construct infinitely many fuzzy spaces, only fuzzy $S^7$ has so far been worked out explicitly, and no concrete applications of these constructions have yet been developed. Further investigation of their mathematical and physical applications is therefore left for future work.


%
\section*{Acknowledgements}
\noindent 
The author was supported by JSPS KAKENHI Grant Number 26K06814.
The author also thanks the participants in the workshop
``Discrete Approaches to the Dynamics of Fields and Space-Time" 
for their useful comments.
In this paper, AI is used for Japanese-to-English translation, as well as for checking for typos, misprints, and errors.\\



\noindent
{\bf Data availability} \   Data sharing is not applicable to this article as no new data were 
created or analyzed in this study.

\section*{Declarations}

{\bf Conflicts of interest} \ 
The author declares that there is no conflict of interest.

\appendix

\section{Reducible representations and geometries} \label{App_red_rep}
We comment on the algebraic variety associated with a reducible representation, which was mentioned in this paper to give a quantization.\\

We first comment from the viewpoint of quantization, that is, from the viewpoint of assigning a noncommutative algebra to a classical algebraic variety. We construct the quantization as a map from a single commutative algebra, namely the coordinate algebra of a single algebraic variety, to a direct product of algebras associated with a reducible representation of the Lie algebra.
Suppose that the representation space is decomposed into the direct sum of irreducible representation spaces as
$V^\mu =\bigoplus_{a=1}^{m_\mu} V_\mu ^a$.
If the quantization maps defined by the irreducible representations are denoted by $q_\mu^a$, then the quantization map associated with the reducible representation is given by
$q^R_\mu:=q^1_\mu\oplus \cdots \oplus q^{m_\mu}_\mu$.
Each irreducible representation appearing in the construction of the direct product algebra can be interpreted as corresponding to a coadjoint orbit, which is a subvariety of the Euclidean space ${\mathbb R}^d$ equipped with the structure of a Lie-Poisson algebra. 
Therefore, choosing suitable sequence of irreducible representations, the image of each map $q_\mu^a$ gives the matrix algebra corresponding to the quantization of the associated coadjoint orbit.\\

Conversely, we comment from the viewpoint of the classical limit on how the noncommutative algebra generated by a reducible representation yields a commutative algebra in the commutative limit.
For a direct product of commutative rings, we have
\begin{align}
Spec(A \times B) = Spec A \coprod Spec B .
\end{align}
In the reducible representation considered in our previous paper \cite{Gohara_Sako}, we observed that the image of
$q^R_\mu:=q^1_\mu\oplus \cdots \oplus q^{m_\mu}_\mu$ generates the algebra $T_\mu:=T^1_\mu\oplus \cdots \oplus T^{m_\mu}_\mu$.
At least in the finite-dimensional case, the direct sum and the direct product can be identified. Therefore, the commutative limit can be interpreted as the disjoint union of the varieties corresponding to the commutative limits of the individual algebras.
Through this correspondence, we see in this paper that, by taking a suitable family of reducible representations, the commutative limit recovers the original algebraic variety before quantization, in the sense that the corresponding components become dense in it.


\section{Notation for ${su}(3)$}
\label{Ap_su(3)}

In this appendix, we summarize the notation for ${su}(3)$ used throughout this paper.
In the fundamental representation, the generators $T_i$ of ${su}(3)$ are given by
\begin{align*}
T_i=\frac{1}{2}\lambda_i,
\end{align*}
where $\lambda_i$ are the Gell-Mann matrices,
\begin{align*}
\lambda_1 &= \begin{pmatrix}
0 & 1 & 0 \\
1 & 0 & 0 \\
0 & 0 & 0
\end{pmatrix}, \quad
\lambda_2 = \begin{pmatrix}
0 & -i & 0 \\
i & 0 & 0 \\
0 & 0 & 0
\end{pmatrix}, \quad
\lambda_3 = \begin{pmatrix}
1 & 0 & 0 \\
0 & -1 & 0 \\
0 & 0 & 0
\end{pmatrix},\\
\lambda_4 &= \begin{pmatrix}
0 & 0 & 1 \\
0 & 0 & 0 \\
1 & 0 & 0
\end{pmatrix}, \quad
\lambda_5 = \begin{pmatrix}
0 & 0 & -i \\
0 & 0 & 0 \\
i & 0 & 0
\end{pmatrix}, \quad
\lambda_6 = \begin{pmatrix}
0 & 0 & 0 \\
0 & 0 & 1 \\
0 & 1 & 0
\end{pmatrix},\\
\lambda_7 &= \begin{pmatrix}
0 & 0 & 0 \\
0 & 0 & -i \\
0 & i & 0
\end{pmatrix}, \quad
\lambda_8 = \frac{1}{\sqrt{3}} \begin{pmatrix}
1 & 0 & 0 \\
0 & 1 & 0 \\
0 & 0 & -2
\end{pmatrix}.
\end{align*}
With respect to this basis, the structure constants are defined by
\begin{align}
i f_{ab}{}^c
=
2\,\mathrm{tr}\bigl([T_a,T_b]T_c\bigr).
\end{align}
Here, since the basis is chosen to consist of Hermitian matrices, as is customary in physics, the structure constants differ by a factor of $i=\sqrt{-1}$.
We use this convention for the structure constants throughout what follows.
The corresponding Lie-Poisson bracket is given by
\begin{align*}
\{x_a,x_b\}
=
f_{ab}{}^c x_c,
\end{align*}
where $x_a$, $a=1,\ldots,8$, are the coordinate functions on
${su}(3)^*$.

Setting
$
X=x_iT_i,
$
we obtain the quadratic Casimir polynomial
\begin{align}
C^2(x)
:=
2\,\mathrm{tr}(X^2)
=
\delta^{ab}x_ax_b.
\label{su(3)_C2}
\end{align}
A cubic Casimir polynomial is given by
\begin{align}
C^3(x)
:=
4\,\mathrm{tr}(X^3)
=&
-\frac{1}{6}
\bigl(
2\sqrt{3}\,x_8^3
-6\sqrt{3}\,x_1^2x_8
-6\sqrt{3}\,x_2^2x_8
-6\sqrt{3}\,x_3^2x_8
+3\sqrt{3}\,x_4^2x_8
\notag\\
&\quad
+3\sqrt{3}\,x_5^2x_8
+3\sqrt{3}\,x_6^2x_8
+3\sqrt{3}\,x_7^2x_8
-18x_2x_5x_6
+18x_2x_4x_7
\notag\\
&\quad
-18x_1(x_4x_6+x_5x_7)
-9x_3\left(x_4^2+x_5^2-x_6^2-x_7^2\right)
\bigr).
\label{su(3)_C3}
\end{align}
The overall normalization of each Casimir polynomial is conventional and is therefore not unique.

The algebraic sets corresponding to quotient algebras of the form
$A_{{su}(3)}/I(C)$ include $\mathbb{R}^8$ itself, the level sets
$
C^2(x)=c_2
$
in $\mathbb{R}^8$, which are spheres $S^7$ when $c_2>0$, and the
level sets
$
C^3(x)=c_3
$
in $\mathbb{R}^8$.\\

Since ${su}(3)$ has rank two, its finite-dimensional irreducible representations are specified by Dynkin labels $(p,q)$, where
$p,q\in\mathbb{Z}_{\geq 0}$. Denoting the fundamental weights by
$\omega_1$ and $\omega_2$, the corresponding highest weight is written as
\begin{align*}
\mu=p\omega_1+q\omega_2.
\end{align*}
By the Weyl dimension formula, the dimension of the representation space
$V^{(p,q)}$ is
\begin{align}
\dim V^{(p,q)}
=
\frac{1}{2}(p+1)(q+1)(p+q+2).
\label{su3_dim}
\end{align}

Let $X_i$, $i=1,\ldots,8$, denote the representation matrices of an orthogonal basis, normalized so that the eigenvalue of the quadratic Casimir operator is
\begin{align}
C_2(p,q)
=
\frac{1}{3}(p^2+q^2+pq)+(p+q).
\label{su3_2nd_Casimir}
\end{align}
Thus,
\begin{align}
\sum_i X_iX_i
=
C_2(p,q)\operatorname{Id}_{V^{(p,q)}}.
\end{align}
The overall coefficient in this formula depends on the normalization of the generators.

With the same normalization, the eigenvalue of the cubic Casimir operator in the representation $(p,q)$ is
\begin{align}
C_3(p,q)
=
\frac{1}{18}(p-q)(2p+q+3)(p+2q+3).
\label{su3_3rd_Casimir}
\end{align}
Again, the overall coefficient depends on the normalization and is not essential to the arguments below.



\end{document}